\documentclass[journal]{IEEEtran}

\usepackage{cite}
\usepackage{graphicx}
\usepackage{wrapfig} 
\usepackage{epsfig}
\usepackage{amsmath}
\usepackage{amsfonts}
\usepackage{subcaption}
\usepackage{psfrag}
\usepackage{rotating}
\usepackage{ifpdf}
\usepackage{latexsym}
\usepackage{dsfont}

\usepackage{stmaryrd}
\SetSymbolFont{stmry}{bold}{U}{stmry}{m}{n} 
\usepackage{amsthm}
\usepackage{amsfonts}
\usepackage{amsmath}
\usepackage{amssymb} 
\usepackage{amsmath}
\usepackage{fancyvrb}
\usepackage{textcomp}
\usepackage{wrapfig}
\usepackage{xfrac} 
\usepackage{stfloats}
\usepackage{url}
\usepackage{ifdraft}
\usepackage{url}
\usepackage{multirow}
\usepackage{multicol}
\usepackage{hyperref}
\usepackage{rotating}
\usepackage{xspace}
\usepackage{array}
\usepackage[usenames]{color} 
\usepackage{arydshln}
\usepackage{multido}
\usepackage{enumerate}
\ifCLASSOPTIONtwocolumn
\usepackage{flushend} 
\fi
\usepackage[latin1]{inputenc}
\usepackage{etoolbox}   
\usepackage[printonlyused]{acronym} 

\theoremstyle{plain}
\newtheorem{thmcounter}{Theorem}
\newtheorem{theorem}[thmcounter]{Theorem}
\newcommand{\theoremref}[1]{Theorem~\protect\ref{#1}}
\theoremstyle{plain}

\newcommand{\corollaryref}[1]{Corollary~\protect\ref{#1}}
\theoremstyle{plain}

\theoremstyle{plain}
\newtheorem{definitioncounter}{Theorem}
\newtheorem{definition}[definitioncounter]{Definition}
\newcommand{\defref}[1]{Definition~\protect\ref{#1}}
\newcommand{\QEDsymbol}{$\hfill\square$}
\newcommand\trigeq{\mathrel{\overset{\makebox[0pt]{\mbox{\normalfont\tiny\sffamily$\!\triangle\!$}}}{=}}}

\newcommand{\figref}[1]{Fig.\,\protect\ref{#1}}

\newcommand{\secref}[1]{Section\,\protect\ref{#1}}

\newcommand{\footnoteref}[1]{\textsuperscript{\rm\protect\ref{#1}}}
\DeclareMathOperator{\mmse}{{\rm{mmse}}}
\def\suscript(#1,#2,#3){{#1}^{#2}_{#3}}
\newcommand{\scalem}[2]{\text{\scalebox{#1}{${#2}$}}}

\newcommand{\bigO}[1]{\mathit{O}\bigl(#1\bigr)}
\newcommand{\gammabar}{{\bar{\gamma}}}
\newcommand{\gammatilde}{\widetilde{\gamma}}
\newcommand{\fracparams}[2]{\genfrac{}{}{0pt}{}{{#1}}{{#2}}}
\newcommand{\emptycoefficient}{\raisebox{0.75mm}{\rule{16pt}{0.4pt}}}
\newcommand{\FoxHDefinition}[3]{\suscript(\rm{H},{#1},{#2}){\left[#3\right]}}
\newcommand{\FoxH}[6][right]{\ifthenelse{\equal{#1}{right}}{\suscript(\rm{H},{#2},{#3}){\left[{#4}\left|\fracparams{#5}{#6}\right.\right]}}{\ifthenelse{\equal{#1}{left}}{\suscript(\rm{H},{#2},{#3}){\left[\left.{#4}\right|\fracparams{#5}{#6}\right]}}{\suscript(\rm{H},{#2},{#3}){\left[{#4}\left|\fracparams{#5}{#6}\right.\right]}}}}
\newcommand{\MeijerGDefinition}[3]{\suscript(\rm{G},{#1},{#2}){\left[#3\right]}}
\newcommand{\MeijerG}[6][right]{\ifthenelse{\equal{#1}{right}}{\suscript(\rm{G},{#2},{#3}){\left[{#4}\left|\fracparams{#5}{#6}\right.\right]}}{\ifthenelse{\equal{#1}{left}}{\suscript(\rm{G},{#2},{#3}){\left[\left.{#4}\right|\fracparams{#5}{#6}\right]}}{\suscript(\rm{G},{#2},{#3}){\left[{#4}\left|\fracparams{#5}{#6}\right.\right]}}}}
\newcommand{\Hypergeom}[5]{\suscript({},{},{#1})\!\,\!\,\suscript({\mathord{F}},{},{#2})\!\left[\scalem{0.95}{{#3};{#4};{#5}}\right]}
\newcommand{\HeavisideTheta}[1]{{\,\mathord{\theta}\!}\left({#1}\right)}
\newcommand{\DiracDelta}[1]{{\,\mathord{\delta}\!\left({#1}\right)}}
\newcommand{\RealPart}[1]{\Re\left\{{#1}\right\}}
\newcommand{\ImagPart}[1]{\Im\!\left\{{#1}\right\}}
\newcommand{\imaginary}{{\rm{i}}}
\newcommand{\abs}[1]{\left|{#1}\right|}

\newcommand{\defrmat}[1]{{\boldsymbol{\MakeUppercase{#1}}}}
\newcommand{\DawsonF}[1]{{\rm{daw}}\!\left({#1}\right)}

\newcommand{\ExpIntegralE}[2]{{\rm{E}}_{{#1}}\!\left({#2}\right)}
\newcommand{\Expected}[1]{{\,{\mathbb{E}}\!\left[{#1}\right]}}
\newcommand{\FourierTransform}[4][norm]{\ifthenelse{\equal{#1}{norm}}{{\mathfrak{F}_{#2}\!\left\{{#3}\right\}\!\left({#4}\right)}}{\ifthenelse{\equal{#1}{conj}}{{\mathfrak{F}_{#2}^{*}\!\left\{{#3}\right\}\!\left({#4}\right)}}{{\mathfrak{F}_{#2}\!\left\{{#3}\right\}\!\left({#4}\right)}}}}
\newcommand{\InvFourierTransform}[3]{{\mathfrak{F}_{#1}^{-{1}}\!\left\{{#2}\right\}\!\left({#3}\right)}}
\newcommand{\MellinTransform}[3]{{\mathfrak{M}_{#1}\!\left\{{#2}\right\}\!\left({#3}\right)}}
\newcommand{\InvMellinTransform}[3]{{\mathfrak{M}_{#1}^{-{1}}\!\left\{{#2}\right\}\!\left({#3}\right)}}
\newcommand{\mathcaltilde}[1]{\widetilde{\mathcal{#1}}}

\newcommand{\manualtimesbreak}{\times\\}

\newcommand{\Mathematica}{{\textsc{\small{Mathematica}}}\textsuperscript{\textregistered}}
\newcommand{\Matlab}{{\textsc{\small{Matlab}}}\textsuperscript{\texttrademark}}
\newcommand{\Maple}{\textsc{\small{Maple}}\textsuperscript{\textregistered}}
\newcounter{acroplace}
\AtBeginEnvironment{acronym}{\setcounter{acroplace}{1}}
\AtEndEnvironment{acronym}{\setcounter{acroplace}{0}}
\newrobustcmd{\acrochoice}[2]{\ifthenelse{\equal{\value{acroplace}}{0}}{\!\!{\itshape{#2}}\ignorespaces}{\ifthenelse{\equal{\value{acroplace}}{1}}{#1}{{\MakeUppercase#2}}}}

\newcommand{\IEEEskipspace}{\vspace{-4mm}}

\hypersetup{
    unicode=false,          
    pdftoolbar=true,        
    pdfmenubar=true,        
    pdffitwindow=false,     
    pdfstartview={FitH},    
    pdftitle={The Title of the Paper},
    pdfauthor={Dr. Ferkan Yilmaz},
    pdfsubject={Subject},   
    pdfcreator={Creator},   
    pdfproducer={latex.exe + dvips.exe + ps2pdf.exe}, 
    pdfkeywords={XXXX, XXXX, XXXX, XXXX.}, 
    pdfnewwindow=true,      
    linkcolor=red,          
    citecolor=green,        
    filecolor=magenta,      
    urlcolor=cyan           
}

\begin{document}

\title{On the Relationships Between Average Channel Capacity, Average Bit Error Rate, Outage Probability and Outage Capacity over Additive White Gaussian Noise Channels}
\author{Ferkan~Yilmaz,~\IEEEmembership{Member,~IEEE,}%
\thanks{Paper approved by ?. ?????, the Editor for ?????????????? ??????~?~??????? ?  of the IEEE ?????????????? ???????. Manuscript received ???? ?, ????; revised ????????? ??, ????. This work was partially supported by Y{\i}ld{\i}z Technical University (YTU).}
\thanks{F.~Yilmaz is with Y{\i}ld{\i}z Technical University, Davutpa\c{s}a Campus, Faculty of Electrical \& Electronics, Department of Computer Engineering, 34220 Esenler, Istanbul, Turkey (e-mail:~ferkan@ce.yildiz.edu.tr).}
}

\ifCLASSOPTIONtwocolumn
\markboth{F. Yilmaz: On the Relationships Between Average Channel Capacity, Average Bit Error Rate, Outage probability and Outage \ldots}{IEEE Transactions on Communications,~Vol.~XX, No.~X, July~2019}
\else
\markboth{F. Yilmaz: On the Relationships Between Average Channel Capacity, Average Bit Error Rate, Outage \ldots}{IEEE Transactions on Communications,~Vol.~XX, No.~X, July~2019}
\fi

\IEEEoverridecommandlockouts
\IEEEpubid{0000--0000/00\$00.00~\copyright~2019 IEEE}


\maketitle

\begin{abstract}
In the theory of wireless communications, average performance measures (APMs) are widely utilized to quantify the performance gains\,/\,impairments in various fading environments under various scenarios, and to comprehend how the factors arising from design/implementation affect system performance. To the best of our knowledge, it has not been yet discovered in the literature how these APMs relate to each other. In this article, having been inspired by the work of Verdu et al.\cite{BibDongningGuoShamaiVerduTIT2005}, we propose that one APM can be calculated using the other APMs instead of using the end-to-end SNR distribution. Particularly, using the Lamperti's transformation (LT), we propose a tractable approach, which we call LT-based APM analysis, to identify a relationship between any two given APMs such that it is irrespective of SNR distribution. Thereby, we introduce some novel relationships among average channel capacity (ACC), average bit error rate (ABER) and outage probability\,/\,capacity (OP\,/\,OC) performances, and accordingly present how to obtain ACC from ABER performance and how to obtain OP\,/\,OC from ACC performance in fading environments. We demonstrate that the ACC of any communications system can be evaluated empirically without using end-to-end SNR distribution. We consider some numerical examples and simulations to validate our newly derived relationships.
\end{abstract}


%
\IEEEpeerreviewmaketitle
\IEEEskipspace

\section*{List of Acronyms}
\ifCLASSOPTIONonecolumn
\begin{multicols}{2}
\fi
\acresetall
\begin{acronym}[HOACC]
\acro{ABER}[ABER]{\acrochoice{Average Bit Error Rate}{average bit error rate}}
\acro{ACC}[ACC]{\acrochoice{Average Channel Capacity}{average channel capacity}}
\acro{ACR}[ACR]{\acrochoice{Average Channel Reliability}{average channel reliability}}
\acro{APM}[APM]{\acrochoice{Average Performance Measure}{average performance measure}}
\acro{AWGN}[AWGN]{\acrochoice{Additive White Gaussian Noise}{additive white Gaussian noise}}
\acro{BER}[BER]{\acrochoice{Bit Error Rate}{bit error rate}}
\acro{CC}[CC]{\acrochoice{Channel Capacity}{channel capacity}}
\acro{CDF}[CDF]{\acrochoice{Cumulative Distribution Function}{cumulative distribution function}}
\acro{CF}[CF]{\acrochoice{Characteristic Function}{characteristic function}}
\acro{CSI}[CSI]{\acrochoice{Channel-Side Information}{channel-side information}}
\acro{CR}[CR]{\acrochoice{Channel Reliability}{channel reliability}}
\acro{FT}[FT]{\acrochoice{Fourier's Transform}{Fourier's transform}}
\acro{GCQ}[GCQ]{\acrochoice{Gauss-Chebyshev Quadrature}{Gauss-Chebyshev quadrature}}
\acro{GNM}[GNM]{\acrochoice{Generalized Nakagami-$m$}{generalized Nakagami-$m$}}
\acro{HOACC}[HOACC]{\acrochoice{Higher-Order Average Channel Capacity}{higher-order average channel capacity}}
\acro{IBP}[IBP]{\acrochoice{Interpolation-based Prediction}{interpolation-based prediction}}
\acro{IFT}[IFT]{\acrochoice{Inverse Fourier's Transform}{inverse Fourier's transform}}
\acro{ILT}[ILT]{\acrochoice{Inverse Lamperti's Transformation}{inverse Lamperti's transformation}}
\acro{IMT}[IMT]{\acrochoice{Inverse Mellin's Transform}{inverse Mellin's transform}}
\acro{LDS}[LDS]{\acrochoice{Lamperti's Dilation Spectrum}{Lamperti's dilation~spectrum}}
\acro{LT}[LT]{\acrochoice{Lamperti's Transformation}{Lamperti's transformation}}
\acro{MGF}[MGF]{\acrochoice{Moment-Generating Function}{moment-generating function}}
\acro{MMSE}[MMSE]{\acrochoice{Minimum Mean-Square Error}{minimum mean-square error}}
\acro{MT}[MT]{\acrochoice{Mellin's Transform}{Mellin's transform}}
\acro{OC}[OC]{\acrochoice{Outage Capacity}{outage capacity}}
\acro{OP}[OP]{\acrochoice{Outage Probability}{outage probability}}
\acro{PDF}[PDF]{\acrochoice{Probability Density Function}{probability density function}}
\acro{PM}[PM]{\acrochoice{Performance Measure}{performance measure}}
\acro{QBP}[QBP]{\acrochoice{Quadrature-based Prediction}{quadrature-based prediction}}
\acro{SNR}[SNR]{\acrochoice{Signal-to-Noise Ratio}{signal-to-noise ratio}}
\end{acronym}
\ifCLASSOPTIONonecolumn
\end{multicols}
\fi

\ifCLASSOPTIONtwocolumn
\IEEEpubidadjcol 
\IEEEskipspace
\fi

\section{Introduction}\label{Section:Introduction}
\IEEEPARstart{I}{n the theory of wireless communications}, \ac{ABER}, \ac{ACC} and \emph{outage probability\,/\allowbreak\,capacity} (\acused{OP}\ac{OP}\,/\,\acused{OC}\ac{OC}) are some of those \acp{APM} that are commonly utilized to investigate various performance aspects of communications systems. Each \ac{APM} serves a pivotal role not only in~interpreting the discovered techniques for higher data transmission, efficient mobility and lower complexity but also in gaining insight into the requirements for achieving sophisticated usage of~radio~frequencies with a higher quality of service. As such, each \ac{APM} has been used to quantify the performance gains\,/\,impairments under various communication scenarios and to comprehend how factors arising from design/implementation (e.g. channel noise, receiver noise, diversity, interference, multipath shadowing and fading) affect overall system performance\cite{BibProakisBook,BibAlouiniBook,BibGoldsmithBook}. 

In the conceptual interpretations, while each \ac{APM} captures different characteristics of~the~\ac{SNR}~distribution, it appears to be the average value of a~different~non-linear transformation of the \ac{SNR} distribution.
For example, a decrease in \ac{ABER} performance, defined as the average rate of bit errors in the transmission, is interpreted~as~an~increase in transmission efficiency and is as well interpreted~as~an~increase in \ac{ACC} performance which is defined as maximum throughput where information can be transmitted error-freely. Further, \ac{OP} performance is defined as the average rate of the event that the \ac{SNR} distribution falls below a given threshold value while \ac{OC} performance is~also described as the same for the event that the information throughput~is~less than the required threshold information throughput. In the literature, \ac{APM} analyses are typically performed using the \ac{PDF} or the \ac{CDF} of \ac{SNR} distribution. This approach is commonly referred to as the \ac{PDF}-based approach, and it is usually more difficult to evaluate than expected, even in the case of combining diversity signals. Therefore, theoreticians, practitioners and engineers in the field of wireless communications have always been embraced by the notion of expressing certain relationships among \acp{APM} in closed-form analytical expressions that are simple in form and likewise straightforward to evaluate.~For~example, for signalling over \ac{AWGN} channels in fading environments, Simon and Alouini presented in \cite{BibSimonAlouiniProcIEEE1998} a relationship between \ac{ABER} performance and the \ac{MGF} of the \ac{SNR} distribution. This relationship was later called \ac{MGF}-based analysis and has been widely used in the literature of wireless communications \cite[and references therein]{BibAlouiniBook}. Annamalai et al. presented in \cite{BibAnnamalaiTellamburaBhargavaTCOM2005} the so-called \ac{CF}-based approach establishing a relationship between \ac{ABER} and the \ac{CF}~of~\ac{SNR} distribution as a robust alternative approach for \ac{MGF}-based analysis. In addition to those relationships, Yilmaz and Alouini developed some \ac{MGF}-based relationships in \cite{BibYilmazAlouiniICC2012,BibYilmazAlouiniTCOM2012,BibYilmazAlouiniTCOM2012UnifiedExpression} to obtain \ac{ACC} and in \cite{BibYilmazAlouiniWCL2012} to obtain \emph{higher-order} ACC (\acused{HOACC}\ac{HOACC}) by using the \ac{MGF} of \ac{SNR} distribution. For approximate \ac{ACC} and approximate \ac{HOACC} analyses, they also recommended in \cite{BibYilmazAlouiniSPAWC2012,BibYilmazTUBITAK2019} to make use of \ac{SNR} moments.~Unlike~the~relationships noted above, Verd\'{u} et al. proposed~in\cite{BibDongningGuoShamaiVerduTIT2005} a remarkable and elegant relationship for the \ac{ACC} performance using the \ac{MMSE} measure, that is
\begin{equation}\label{Eq:RelationshipBetweenMMSEAndCapacity}
    \mathcal{C}_{avg}(\gammabar)=\int_{0}^{\gammabar}\mmse(u)\,du,
\end{equation}
where $\gammabar$ is the average \ac{SNR}. Further, $\mathcal{C}_{avg}(\gammabar)$ and $\mmse(\gammabar)$ are the performance measures of \ac{ACC} and \ac{MMSE}, respectively. It should be worth noticing that, among~all~the~relationships above, only the one proposed by Verd\'{u} et al. \cite{BibDongningGuoShamaiVerduTIT2005} is noteworthily~not~only irrespective~of~the~end-to-end~\ac{SNR} distribution but is also regardless of the broadest \ac{SNR} settings. As such, it reveals a fundamental connection between information and estimation measures, and distinctively illuminates intimate connections between information theory and estimation theory. It is therefore important in the theory of wireless communications to investigate the existence of such a relationship that is irrespective of the end-to-end \ac{SNR}~distribution.~To~the~best~of our knowledge, the relationship between any two \acp{APM} has not been yet discussed in the literature, nor has an approach been proposed how to establish it theoretically.

\subsection{Our Contributions}
Motivated by the literature mentioned above, we propose~in this article the goal of performing any \ac{APM} analysis~using~the other \acp{APM} instead of using the statistical knowledge (such as \ac{PDF}, \ac{CDF}, \ac{MGF} and higher-order moments) of the end-to-end \ac{SNR} distribution. With that goal, we investigate~whether~a~relationship exists between any given two \acp{APM}, not only irrespective of \ac{SNR} distribution but also applicable to the broadest \ac{SNR} settings of \ac{AWGN} channels. Our main contributions can be summarized as follows.
\begin{itemize}
\setlength\itemsep{1mm}
    \item We first recommend in \secref{Section:RelationshipsAmongAPMs} the \ac{LT}\cite{BibLampertiAMS1962} to identify the similarity between any two \acp{APM} concerning the average SNR dilations, which constitutes behind the novelty of this work. In particular, we define the \ac{LDS} of an \ac{APM} as the Fourier spectrum of its \ac{LT} with respect to the average SNR dilations, and therewith show the existence conditions for the similarity between the \ac{LDS} spectrums.

    \item Secondly, we propose in \theoremref{Theorem:RelationshipBetweenAPMs} a tractable approach, which we call \ac{LT}-based approach for performance analysis, to compute one \ac{APM} using the other \ac{APM}, especially without needing the statistical knowledge (such as \ac{PDF}, \ac{CDF}, \ac{MGF} and moments) of \ac{SNR} distribution. In particular, we show how~to~establish~a~relationship between any two \acp{APM}, and also verify that the relationship we obtain using the \ac{LT}-based approach is irrespective of \ac{SNR} distribution and hence holds under a variety of all \ac{SNR} settings of \ac{AWGN} channels including the discrete-time and continuous-time channels, either in scalar or vector versions. To the best of our knowledge, this \ac{LT}-based approach changes the playground of performance analysis since being based on the similarity between averaged statistics rather than between instantaneous statistics. Using one \ac{APM} either experimentally easy to measure or mathematically simple to derive, we are able to obtain the other \ac{APM} that is either experimentally difficult to measure or mathematically tedious to obtain. Further, in \secref{Section:RelationshipsAmongAPMs}, We investigate the existence of such relationship between any two APMs, and emphasize Mellin's convolution in connection with obtaining closed-form expressions.
    
    \item To the best of our knowledge, the literature has currently no answer on how we determine \ac{ACC} either empirically or experimentally without using~the~statistical~knowledge of SNR distribution. As regards an application~of~our~\ac{LT}-based approach, we propose in \secref{Section:RelationshipBetweenACCandABEP} a relationship between \ac{ACC} and \ac{ABER} of any communications system. With the aid of this relationship, we demonstrate,~for~the first time in the literature, that we are readily able to empirically measure the \ac{ACC} of any communications~system without the need for the statistical knowledge of \ac{SNR} distribution and all the \ac{SNR} settings, that is
    \begin{equation}\label{Eq:RelationshipBetweenBEPAndCapacity} 
        \mathcal{C}_{avg}(\gammabar)=
            \int_{0}^{\infty}{z}(u)\,
                \Bigl\{1-\mathcal{E}_{avg}(u\,\gammabar)\Bigr\}\,du,
    \end{equation}
    where ${z}(u)$ is the auxiliary function that depends on the modulation scheme and defined in \secref{Section:RelationshipBetweenACCandABEP} and $\mathcal{E}_{avg}(\gammabar)$ denotes the \ac{ABER} whose empirical measurement does not require the statistical knowledge of the \ac{SNR} distribution and all the broadest \ac{SNR} settings\footnote{One of the most important questions that arise when we design communications systems is to ask how much information can be error-freely transferred in a given period of time. The answer to this question encourages the usage of \ac{ACC}. However, \ac{ACC} is such a ghost-like criterion and empirically difficult to measure\cite{BibAlouiniBook,BibGoldsmithBook}, requiring the knowledge of \ac{SNR} distribution and the broadest SNR settings because it is defined as \emph{a theoretical upper-bound} to the transmission  throughput\cite{BibShannonBSTJ1948,BibShannon1949,BibShannonWeaverUIP1949}. On the other hand, \ac{ABER} is empirically much simple and less costly to measure without the need for the knowledge of \ac{SNR} distribution and the broadest \ac{SNR} settings. The basic concept is to transmit random equally probable bits over the channel and, after detection, to find the average rate of the received erroneous bits, given the total number of transmitted bits. For example, if 5-bit errors occur during transmission of a million bits, the \ac{ABER} is $5/1000000$ or $5\times{10}^{-6}$.}.
    
    \item The other two \acp{APM} widely~encountered~in~the~theory~of wireless communications are \ac{OP} and \ac{OC} performances. With the aid of our \ac{LT}-based approach, we propose that both \ac{OP} and \ac{OC} of any communications system can be obtained by using its exact \ac{ACC} performance, that are
    \setlength\arraycolsep{1.4pt}
    \begin{eqnarray}
        \mathcal{P}_{out}(\gammabar;\gamma_{th})
                &=&\frac{1}{\pi}
                        \Im\Bigl\{\mathcal{C}_{avg}(-\,\gammabar/\gamma_{th})\Bigl\},\\
        \mathcal{C}_{out}(\gammabar;C_{th})
                &=&\frac{1}{\pi}
                        \Im\Bigl\{\mathcal{C}_{avg}(-\,\gammabar/(e^{C_{th}}-1))\Bigl\},
    \end{eqnarray}
    where $\ImagPart{\cdot}$ denotes the imaginary part of~the~term~enclosed. Further, $\mathcal{P}_{out}(\gammabar;\gamma_{th})$ and $\mathcal{C}_{out}(\gammabar;C_{th})$ denote the \ac{OP} for a certain \ac{SNR} threshold $\gamma_{th}$ and the \ac{OC} for a certain information throughput threshold $C_{th}$, respectively. To the best of our knowledge, this relationship has not been yet reported in the literature. 

    \item The other contribution of this article is that, in contrary~to the approaches, previously reported in the literature, that provide average statistics using sample statistics such as \ac{PDF} and \ac{CDF}, our \ac{LT}-based approach allows us to obtain sample statistics using average statistics. In particular,~we propose that, using the ACC of any information transmission, we can find the PDF of the SNR distribution to which information transfer is theoretically exposed. Let $\gamma$ denote the \ac{SNR} distribution. Then, we propose that its \ac{PDF} $f_{\gamma}(x;\gammabar)$ is given by
    \begin{equation}\label{Eq:RelationshipBetweenPDFAndACC}
        f_{\gamma}(x;\gammabar)=
            \frac{1}{\pi}
                \Im\Bigl\{\frac{\partial}{\partial{x}}\,
                    \mathcal{C}_{avg}(-\,\gammabar/{x})\Bigl\}.
    \end{equation}
    Using this result within the \ac{PDF}-based approach\cite{BibProakisBook,BibAlouiniBook,BibGoldsmithBook}, we introduce \ac{CC}-based performance analysis of wireless communications to perform any \ac{APM} analysis by using the exact \ac{ACC} expressions. 
\end{itemize}

Consequently, viewed from a somewhat broader perspective, the \ac{LT}-based approach opens a set of new ideas and techniques in communications theory.  

\subsection{Article Organization}
The remainder of the article is organized as follows. In~\secref{Section:RelationshipsAmongAPMs}, we introduce the \ac{LT} and the \ac{LDS} of \ac{APM}~and~then propose the \ac{LT}-based approach to establish relationships among \acp{APM}. By means of the \ac{LT}-based approach, we propose~in \secref{Section:RelationshipBetweenACCandABEP} a novel analytical relationship to obtain the \ac{ACC} of any communication system from its \ac{ABER}, and thereafter in \secref{Section:RelationshipBetweenAOPandAOCandACC} another novel relationship to find its \ac{OP} and \ac{OC} performances from its exact \ac{ACC} performance. Finally, our conclusions are drawn in the last section. 
\ifCLASSOPTIONtwocolumn
\pagebreak[4]
\fi

\subsubsection*{Notation} 
The following notations are used in this article.
Scalar numbers such as integer, real and complex numbers are denoted by lowercase letters, e.g. $x$, $y$ and $z$. Let $\mathbb{N}$ denote the set of all natural numbers. Let $\mathbb{R}$ denote the set of all real numbers such that $\mathbb{R}_{-}$ and $\mathbb{R}_{+}$ denote the sets of all negative and all positive real numbers, respectively. Further, let $\mathbb{C}$ denote the set of all complex numbers. If $z\in\mathbb{C}$, then it is written as $z\!=\!{x}+\imaginary{y}$, where $\imaginary\!\trigeq\!\sqrt{-1}$ denotes the imaginary number, and where $x,y\!\in\!\mathbb{R}$ are called the inphase and the quadrature, respectively, such that $x\!=\!\RealPart{z}$ and $y\!=\!\ImagPart{z}$, where where $\Re\{\cdot\}$ and $\Im\{\cdot\}$ give the real part and imaginary part of a given complex number, respectively. Further, the complex conjugate of $z\in\mathbb{C}$ is denoted by $z^{*}\!=\!\RealPart{z}-\imaginary\ImagPart{z}$. 

To make the results of probability and statistics clear and concise, random distributions are denoted by uppercase letters, e.g. $X$, $Y$, $Z$. Random vectors and random matrices will be denoted by calligraphic boldfaced uppercase letters, e.g. $\defrmat{X}$, $\defrmat{Y}$ and $\defrmat{Z}$. Let $X$ be a random distribution, then its \ac{PDF} is defined by 
\begin{equation}
    f_X(x)\!=\!\Expected{\DiracDelta{x-X}},
\end{equation}
where $\Expected{\cdot}$ denotes the expectation operator, and $\DiracDelta{\cdot}$ denotes the Dirac's delta function\cite[Eq.\!~(1.8.1)]{BibZwillingerBook}.
Besides, its \ac{CDF} is defined by
\begin{equation}
    F_X(x)\!=\!\Expected{\HeavisideTheta{x-X}},
\end{equation}
where $\HeavisideTheta{\cdot}$ is the Heaviside's theta function\cite[Eq.\!~(1.8.3)]{BibZwillingerBook}.
Furthermore, the conditional \ac{PDF} and \ac{CDF} of $X$ given $G$ will also be denoted by $f_{X|G}(x|g)$ and $F_{X|G}(x|g)$, respectively. 

\section{Background on the relationships among \acp{APM}}
\label{Section:RelationshipsAmongAPMs}
In wireless communications, many interrelated factors that result not only from design\,/\,implementation of technologies but also from fading conditions influence the \ac{SNR} distribution denoted by $\gamma\colon\mathbb{R}^{L}\!\to\!\mathbb{R}_{+}$, where $L$ denotes the number of \ac{SNR} settings, that is $\gamma\!=\!\gamma\left(\boldsymbol{\Psi}\right)$, where $\boldsymbol{\Psi}\!=\![{\Psi}_{1},{\Psi}_{2},\ldots,{\Psi}_{L}]$ denotes those interrelated factors known as the \ac{SNR} settings of information transmission. The joint \ac{PDF} $\boldsymbol{\Psi}$ is given as $p_{\boldsymbol{\Psi}}(\boldsymbol{\psi})$, where $\boldsymbol{\psi}=[{\psi}_{1},{\psi}_{2},\ldots,{\psi}_{L}]$\footnote{For all $\ell\in\{1,2,\ldots,L\}$, the factor $\Psi_{\ell}$ is, without loss of generality, either a random distribution following the \ac{PDF} $p_{\Psi_{\ell}}(\psi)=\mathbb{E}\left[\DiracDelta{\psi-\Psi_{\ell}}\right]$ or a constant value that could also be viewed as a random distribution following the \ac{PDF} $p_{\Psi_{\ell}}(\psi)=\DiracDelta{\psi-\Psi_{\ell}}$, where $\DiracDelta{\cdot}$ denotes Dirac's delta function\cite[Eq.(1.8.1)]{BibZwillingerBook}. Further, when  $\Psi_{1},\Psi_{2},\ldots,\Psi_{L}$ are mutually independent, their joint \ac{PDF} can be written as $p_{\mathbf{\Psi}}(\boldsymbol{\psi})=\prod_{\ell=1}^{L}p_{{\Psi}_{\ell}}({\psi}_{\ell})$.} without loss of generality. The \acp{PM} capture different characteristics of \ac{SNR} distribution and enhance the understanding of physical and practical scenarios. Let us choose two different \acp{PM}, i.e., $\mathcal{G}\left(\gamma\right)$ and $\mathcal{H}\left(\gamma\right)$, each of which is certainly continuous and a monotonically increasing or decreasing function.\footnote{If $f\left(x\right)$ is completely monotonic, then $f^{(n)}\left(x\right)={(\partial/\partial{x})}^{n}f\left(x\right)$ exists everywhere such that $(-1)^{n}f^{(n)}\left(x\right)\geq{0}$ for all ${n}\in\mathbb{Z}_{+}$.} The \ac{APM} $\mathcal{G}_{avg}(\gammabar)=\mathbb{E}[\mathcal{G}({\gamma}\left(\boldsymbol{\Psi}\right)]$ can be properly written conditioned on the parameters $\boldsymbol{\Psi}$, that is
\begin{equation}
\label{Eq:AveragedPerformance}
\mathcal{G}_{avg}(\gammabar)=\underbrace{\int\int\ldots\int}_{\text{$L$-fold}}
	\mathcal{G}\bigl({\gamma\left(\boldsymbol{\psi}\right)}\bigr)
		p_{\boldsymbol{\Psi}}(\boldsymbol{\psi})\,{d}\boldsymbol{\psi},
\end{equation}
where $\gammabar\trigeq\mathbb{E}[{\gamma\left(\boldsymbol{\Psi}\right)}]\in\mathbb{R}_{+}$ 
denotes the average \ac{SNR} given by
\begin{equation}\label{Eq:AveragedSNR}
\gammabar=\underbrace{\int\int\ldots\int}_{\text{$L$-fold}}
	{\gamma\left(
		\boldsymbol{\Psi}
		\right)}
		p_{\boldsymbol{\Psi}}(\boldsymbol{\psi})\,{d}\boldsymbol{\psi}.
\end{equation}
Accordingly, as for the performance measure $\mathcal{H}(\gamma)$, the \ac{APM}
$\mathcal{H}_{avg}(\gammabar)=\mathbb{E}[\mathcal{H}({\gamma}\left(\boldsymbol{\Psi}\right)]$ can also be rewritten in the form of \eqref{Eq:AveragedPerformance}. Without loss of generality, we suppose that, as compared to $\mathcal{H}_{avg}(\gammabar)$, $\mathcal{G}_{avg}(\gammabar)$ is either mathematically more tractable resulting in closed-form expressions, or numerically more efficient to compute, or experimentally easier to measure. Particularly, for a certain set of average \acp{SNR} $\{\gammabar_{1},\gammabar_{2},\ldots,\gammabar_{N}\}$, $N\in\mathbb{N}$, we can experimentally obtain a measurement set
\begin{equation}\label{Eq:MeasurementSet}
\mathcal{S}_{N}=
    \bigl\{
        (\gammabar_{n},\mathcal{G}_{avg}(\gammabar_{n}))
        \,\bigl|\,
        n=0,1,2,\ldots,N
        \bigr.
    \bigr\}
\end{equation}
From the viewpoints outlined previously, we attempt to define a relationship from $\mathcal{G}_{avg}(\gammabar)$ to $\mathcal{H}_{avg}(\gammabar)$ irrespective of \ac{SNR} distribution. Our intuitive approach is thus search for a linear relationship with an \ac{SNR} invariant filter, that is
\begin{equation}\label{Eq:DesiredRelationAmongAPMs}
    \mathcal{H}_{avg}(\gammabar)=
	    \mathcal{Z}_{N}(\mathcal{G}_{avg}\bigl(\gammabar_{1}),\mathcal{G}_{avg}(\gammabar_{2}),\ldots,\mathcal{G}_{avg}(\gammabar_{N})\bigr),
\end{equation}
where $\mathcal{Z}_{N}(g_1,g_2,\ldots,g_N)$ is the auxiliary function required to establish the relationship, and in general, rewritten as a multivariate linear function $\mathcal{Z}_{N}(g_1,g_2,\ldots,g_N)\!=\!\sum_{n=1}^{N}z_{n}\,g_{n}$, where $z_{1},z_{2},\ldots,z_{N}$ are required. Placing this linear function into \eqref{Eq:DesiredRelationAmongAPMs}, and therein without loss of generality, interpreting each average \ac{SNR} as a dilation of $\gammabar$ (i.e., $\forall{n}\in\{1,2,\ldots,N\}$, $\gammabar_{n}=\lambda_{n}\gammabar$ with a certain dilation $\lambda_{n}\in\mathbb{R}_{+}$), we have
\begin{equation}\label{Eq:DilationBasedRelationshipAmongAPMs}
    \mathcal{H}_{avg}(\gammabar)=
        \sum_{n=1}^{N}z_{n}\,\mathcal{G}_{avg}(\lambda_{n}\,\gammabar),
\end{equation}
which suggests the self-similarity (or scale invariance) of $\mathcal{H}_{avg}(\gammabar)$ under positive $\lambda_1,\lambda_2,\ldots,\lambda_N$ dilations. Once acknowledged as an important feature \cite{BibMandelbrotBook,BibLampertiAMS1962}, scale-invariance is indeed a fundamental property for physical phenomena. 

\begin{definition}[Dilation operator]\label{Def:DilationOperator}
Let $\mathcal{X}(\gammabar)$ be an \ac{APM} measured for a specific average \ac{SNR} $\gammabar$. The dilation operator $\mathfrak{D}_{H,\gammabar}\bigl\{\cdot\bigr\}(\lambda)$ is defined as
\begin{equation}\label{Eq:NormalizedDilationOperator}
\mathfrak{D}_{H,\gammabar}
	\bigl\{\mathcal{X}\left(\gammabar\right)\bigr\}(\lambda)
		\trigeq\lambda^{H}\,\mathcal{X}(\lambda\,\gammabar),
\end{equation}
where $\lambda\in\mathbb{R}_{+}$ and $H\in\mathbb{R}$ denote the dilation and the Hurst exponent, respectively.\QEDsymbol
\end{definition}

Accordingly, the performance $\mathcal{H}_{avg}(\gammabar)$ is said to be self-similar (or scale invariant) with a specific scaling exponent $H$ if and only if the following condition is provided, that is
\begin{equation}\label{Eq:ScaleInvarianceCondition}
\mathfrak{D}_{H,\gammabar}
	\bigl\{\mathcal{H}_{avg}\left(\gammabar\right)\bigr\}(\lambda)
		\trigeq\mathcal{H}_{avg}(\gammabar)~\text{for all $\lambda\in\mathbb{R}_{+}$.}
\end{equation}
Although the property of self-similarity (scale invariance) is quite convenient to interpret physical phenomena, it has not attracted much attention in the literature of wireless communications theory. For example, the fractional moments of \ac{SNR} distribution (i.e., the power fluctuations of the additive noise) is scale invariant with respect to the average \ac{SNR}.\!\footnote{\label{Footnote:ScaleInvariantFractionalMoment} 
Let $\mu_{\gamma}(n;\gammabar)$ be the $n$th moment of the \ac{SNR} ${\gamma}$, i.e., $\mu_{\gamma}(n;\gammabar)=\mathbb{E}[\gamma^{n}]$ where $n\in\mathbb{R}_{+}$. Accordingly, using $\mu_{\gamma}(n;\lambda\,\gammabar)=\lambda^{n}\,\mu_{\gamma}(n;\gammabar)$, the fractional moments of \ac{SNR} can be shown to be scale invariant with respect to the average \ac{SNR} $\gammabar$, that is,
\begin{equation}
\mathfrak{D}_{H,\gammabar}
	\bigl\{\mu_{\gamma}(n;\gammabar)\bigr\}(\lambda)=
		\mu_{\gamma}(n;\gammabar)
		\tag{F.\thefootnote.1}
\end{equation}
with Hurst exponent $H=-n$.} However, since the scale invariance property described in \eqref{Eq:ScaleInvarianceCondition} does not hold for all \acp{APM}, we consider benefiting from the \ac{LT} to extend the dilation to an exponential dilatation. It is further worth mentioning that the \ac{LT} enables us to utilize the nature of the incremental Gaussian channel that contains cascading channels whose average \acp{SNR} decreases along their order\cite{BibDongningGuoShamaiVerduTIT2005}.

\begin{definition}[Lamperti's transformation]\label{Def:LimpertiTransform}
Let $\mathcal{X}(\gammabar)$ be an \ac{APM} measured for an average \ac{SNR} of $\gammabar$. Lamperti's inverse transformation $\mathfrak{L}^{-1}_{H,\gammabar}\bigl\{\cdot\bigr\}(\lambda)$ exercised on $\mathcal{X}(\gammabar)$ is given by
\setlength\arraycolsep{1.4pt}
\begin{equation}\label{Eq:LimpertiTransfom}
\mathfrak{L}^{-1}_{H,\gammabar}\Bigl\{\mathcal{X}\left(\gammabar\right)\Bigr\}(\lambda)
	\trigeq{e}^{H\lambda}\mathcal{X}(e^{-\lambda}\,\gammabar)
	    =\mathcal{L}_{\mathcal{X}}\!\left(\lambda,\gammabar\right),
\end{equation}
for all $\lambda\in\mathbb{R}_{+}$, where the direct transformation $\mathfrak{L}_{H,\lambda}\bigl\{\cdot\bigr\}(\gammabar)$ is called Lamperti's transformation defined by
\begin{equation}\label{Eq:LimpertiInverseTransfom}
\mathfrak{L}_{H,\lambda}\Bigl\{\mathcal{L}_{\mathcal{X}}\!\left(\lambda,\gammabar\right)\Bigr\}(\gammabar)
	\trigeq\gammabar^{H}\mathcal{L}_\mathcal{X}\left(-\log(\gammabar),1\right)
		=\mathcal{X}\!\left(\gammabar\right),
\end{equation}
such that $\mathfrak{L}_{H,\lambda}\bigl\{\mathfrak{L}^{-1}_{H,\gammabar}\bigl\{\mathcal{X}\left(\gammabar\right)\bigr\}(\lambda)\bigr\}(\gammabar)\trigeq\mathcal{X}(\gammabar)$.
\QEDsymbol
\end{definition}

The \ac{ILT} of an \ac{APM}, once explained in what follows, has some useful properties in Fourier domain, namely the fact that its Fourier spectrum remains constant for any dilation of the average \ac{SNR}. Hence, these constant quantities defines the similarities among \acp{APM}, verily allowing ome \ac{APM} to be estimated using the other \ac{APM}. Appropriately, we obtain the \ac{ILT} of $\mathcal{H}_{avg}(\gammabar)$ as 
\begin{equation}
    \mathfrak{L}^{-1}_{H,\gammabar}\Bigl\{\mathcal{H}_{avg}(\gammabar)\Bigr\}(\lambda)    \trigeq{e}^{H\lambda}\mathcal{H}_{avg}\bigl(e^{-\lambda}\gammabar\bigr)
\end{equation}
whose Fourier spectrum with respect to the dilation $\lambda\in\mathbb{R}_{+}$ is called the \ac{LDS} and invariant with respect to the dilation of $\gammabar$. Taking \eqref{Eq:DilationBasedRelationshipAmongAPMs} into consideration, we deduce that an analytical relationship between two \ac{APM} could be established involving the dilation of $\gammabar$. 

\begin{theorem}[Lamperti's dilation spectrum]\label{Theorem:ScaleInvariantSpectrum}
Let $\mathcal{X}(\gammabar)$ be an \ac{APM}. Then, the \ac{LDS} of $\mathcal{X}(\gammabar)$ is defined as 
\begin{equation}\label{Eq:ScaleInvariantSpectrum}
	\mathcaltilde{X}(\omega,\gammabar)
		\trigeq\FourierTransform{\lambda}{\mathfrak{L}^{-1}_{H,\gammabar}\Bigl\{\mathcal{X}\left(\gammabar\right)\Bigr\}(\lambda)}{\omega},
		\quad\omega\in\mathbb{R},
\end{equation}
where $\FourierTransform{\cdot}{\cdot}{\cdot}$ denotes the \ac{FT}\footnote{\label{Footnote:FourierTransform}
Let $\phi\colon\mathbb{R}\to\mathbb{R}$ be a real-valued and monotonic function locally integrable and differentiable. It is further assumed that $\phi(\gamma)$ is continuous over $\gamma\in\mathbb{R}$. Then, the \acf{FT} is defined as
\begin{equation}\label{FootEq:FourierTransform}
	\Phi(\omega)=
		\FourierTransform{\gamma}{\phi(\gamma)}{\omega}
			\trigeq\int_{-\infty}^{+\infty}{e}^{\imaginary\omega\gamma}\,\phi(\gamma)\,d\gamma,
				\tag{F.\thefootnote.1}
\end{equation}
where $\omega\in\mathbb{R}$ denotes the frequency scale and $\imaginary=\sqrt{-1}$ denotes imaginary number, and where the sufficient condition for the existence of \ac{FT} is given by
$\int_{-\infty}^{+\infty}\abs{\phi(\gamma)}^{2}\,d\gamma<\infty$. Further, the \acf{IFT} of $\Phi(\omega)$ is defined as
\begin{equation}\label{FootEq:InverseFourierTransform}
	\phi(\gamma)=
		\InvFourierTransform{\omega}{\Phi(\omega)}{\gamma}
			\trigeq\frac{1}{2\pi}\int_{-\infty}^{\infty}\Phi(\omega)\,{e}^{-\imaginary\omega\gamma}\,d\omega.
				\tag{F.\thefootnote.2}
\end{equation}
For more information, the readers and researchers are referred to \cite{BibChampeneyBook,BibSneddonFourierTrasforms1995Book}.
}. $\mathcaltilde{X}(\omega,\gammabar)$ is scale
invariant with respect to the average \ac{SNR} $\gammabar$, i.e., 
\begin{equation}\label{Eq:ScaleInvariantSpectrumDilation}
	\mathfrak{D}_{H+\imaginary\,\omega,\gammabar}
		\bigl\{\mathcaltilde{X}(\omega,\gammabar)\bigr\}(\lambda)
			=\mathcaltilde{X}(\omega,\gammabar),\quad{\lambda}\in\mathbb{R}_{+}.
\end{equation}
where $\imaginary=\sqrt{-1}$ denotes imaginary number.
\end{theorem}

\begin{proof}
In accordance with the definition of the \ac{FT}, the spectrum $\mathcaltilde{X}(\omega,\gammabar)$, given in \eqref{Eq:ScaleInvariantSpectrum}, can be written as\footnoteref{Footnote:FourierTransform}
\begin{subequations}
\label{Eq:LimpertiTransformBasedSpectrum}
\setlength\arraycolsep{1.4pt}
\begin{eqnarray}
    \label{Eq:LimpertiTransformBasedSpectrumA}
    \mathcaltilde{X}(\omega,\gammabar)
        &\trigeq&\int_{-\infty}^{+\infty}
            e^{\imaginary\,\omega\lambda}\,
                \mathfrak{L}^{-1}_{H,\gammabar}\Bigl\{\mathcal{X}\!\left(\gammabar\right)\Bigr\}(\lambda)\,d\lambda,\\
    \label{Eq:LimpertiTransformBasedSpectrumB}
        &=&\int_{-\infty}^{+\infty}
            e^{(H+\imaginary\,\omega)\lambda}\,
                \mathcal{X}\!\left(e^{-\lambda}\,\gammabar\right)\,d\lambda,
\end{eqnarray}
\end{subequations}
where the Hurst exponent $H\in\mathbb{R}$ has to be suitably~and~carefully chosen in such a way, which is explained in the following theorem, to guarantee the convergence\,/\,existence of the \ac{FT}. Further, changing the variable
$e^{-\lambda}\gammabar\rightarrow{}\beta$ in \eqref{Eq:LimpertiTransformBasedSpectrumB} results in
$\mathcaltilde{X}(\omega,\gammabar)=\gammabar^{H+\imaginary\omega}            \int_{0}^{+\infty}\beta^{-H-1-\imaginary\,\omega}\mathcal{X}\left(\beta\right)d\beta$, where setting $\gammabar\rightarrow{\lambda}\gammabar$ yields $\mathcaltilde{X}(\omega,\lambda\,\gammabar)={\lambda}^{H+\imaginary\,\omega}\,\mathcaltilde{X}(\omega,\gammabar)$ for all ${\lambda}\in\mathbb{R}_{+}$. Using \defref{Def:DilationOperator}, this result can
be easily simplified to \eqref{Eq:ScaleInvariantSpectrumDilation}, which proves \theoremref{Theorem:ScaleInvariantSpectrum}.
\end{proof}

Since the \ac{FT} is an improper integral, the conditions for the existence of \ac{LDS} are complicated to state in general but are sufficiently given in the following theorem. 

\begin{theorem}[Existence of Lamperti's dilation spectrum]
\label{Theorem:LDSExistence} 
Let $\mathcal{X}(\gammabar)$ be an \ac{APM}. Iff
\setlength\arraycolsep{1.4pt}
\begin{eqnarray}
    \mathcal{X}(\gammabar)&=&\bigO{\gammabar^{-\alpha}}{~}\text{for}{~}\gammabar\rightarrow{0^{+}},\\
    \mathcal{X}(\gammabar)&=&\bigO{\gammabar^{-\beta}}{~}\text{for}{~}\gammabar\rightarrow{+\infty},
\end{eqnarray}
such that $\alpha<\beta$, then the \ac{LDS} $\mathcaltilde{X}\left(\omega,\gammabar\right)$ exists for any Hurst exponent $H\in(\alpha,\beta)$.
\end{theorem}

\begin{proof}
Let $\mathcal{L}_{\mathcal{X}}\!(\lambda,\gammabar)$ be the \ac{ILT} of $\mathcal{X}(\gammabar)$, i.e., $\mathcal{L}_{\mathcal{X}}\!(\lambda,\gammabar)\equiv\mathfrak{L}^{-1}_{H,\gammabar}\bigl\{\mathcal{X}(\gammabar)\bigr\}(\lambda)$, and suppose that
$\int_{0}^{\lambda}\bigl|\mathcal{L}_{\mathcal{X}}\!(u,\gammabar)\bigr|du\!<\!\infty$ for any finite dilation $\lambda\in\mathbb{R}_{+}$. As per the existence conditions of \ac{FT} \cite{BibChampeneyBook}, whenever $\mathcal{L}_{\mathcal{X}}\!(\lambda,\gammabar)$ is of exponential order, its \ac{FT} certainly exists, that is,
\begin{equation}\label{Eq:LampertiSpectrumExponentialOrder}
\mathcal{L}_{\mathcal{X}}\left(\lambda,\gammabar\right)=
	\left\{{
		\setlength\arraycolsep{1.4pt}
		\begin{array}{lr}
			e^{H\lambda}\bigO{e^{-\alpha\lambda}}&\text{as $\lambda\rightarrow{-\infty}$},\\
			e^{H\lambda}\bigO{e^{-\beta\lambda }}&\text{as	$\lambda\rightarrow{+\infty}$},
		\end{array}}
	\right.
\end{equation}
which implies that $\alpha<H$ and $H<\beta$. Thus, $\mathcaltilde{X}\left(\omega,\gammabar\right)$
exists for any Hurst exponent $H\in(\alpha,\beta)$. Applying the \ac{LT} to
\eqref{Eq:LampertiSpectrumExponentialOrder} yields $\mathcal{X}(\gammabar)\!=\!\bigO{\gammabar^{-\alpha}}$ as $\gammabar\!\rightarrow\!{0^{+}}$ and
$\mathcal{X}(\gammabar)\!=\!\bigO{\gammabar^{-\beta}}$ as $\gammabar\!\rightarrow\!{\infty}$, which proves \corollaryref{Theorem:LDSExistence}.
\end{proof}

Note that, noticing the precise description of how the \ac{LDS} changes while from the average \ac{SNR} $\gammabar$ to its dilated version $\lambda\gammabar$, $\lambda\in\mathbb{R}_{+}$, we consider that \theoremref{Theorem:ScaleInvariantSpectrum} is so much beneficial to extract the features of \ac{APM} and especially to disclose the similarities and differences among \acp{APM}. Accordingly, we can establish theoretical relationships between two \acp{APM} using the ratio of their \acp{LDS} given the following \theoremref{Theorem:SimilarityBetweenLDSs}.

\begin{theorem}[Similarity between Lamperti's dilation spectrums]
\label{Theorem:SimilarityBetweenLDSs} 
The two \acp{APM} $\mathcal{X}(\gammabar)$ and $\mathcal{Y}(\gammabar)$ are similar iff their \acp{LDS}  $\mathcaltilde{X}\left(\omega,\gammabar\right)$ and $\mathcaltilde{Y}\left(\omega,\gammabar\right)$ provide  
\begin{equation}\label{Eq:SpectrumSimilarity}
	\frac{\mathcaltilde{X}\left(\omega,\lambda_{1}\gammabar\right)}
		 {\mathcaltilde{Y}\left(\omega,\lambda_{2}\gammabar\right)}
			\trigeq\lambda^{H+\imaginary\,\omega},
\end{equation}
for all $\lambda_1,\lambda_2\in\mathbb{R}_{+}$, where $\lambda=\lambda_2/\lambda_1$.
\end{theorem}

\begin{proof}
The proof is obvious using \theoremref{Theorem:ScaleInvariantSpectrum}.
\end{proof}

In order to establish an analytical relationship between two \acp{APM} $\mathcal{G}_{avg}(\gammabar)$ and $\mathcal{H}_{avg}(\gammabar)$, we need to find out whether there exits a similarity between them by means of \theoremref{Theorem:SimilarityBetweenLDSs}. As performing in accordance with \theoremref{Theorem:ScaleInvariantSpectrum} and \theoremref{Theorem:LDSExistence}, we obtain the \ac{LDS} of $\mathcal{G}_{avg}(\gammabar)$ as
\begin{equation}
\label{Eq:LDSForGavg}
\mathcaltilde{G}_{avg}\left(\omega,\gammabar\right)
	\trigeq\FourierTransform{\lambda}{\mathfrak{L}^{-1}_{H,\gammabar}\bigl\{\mathcal{G}_{avg}\left(\gammabar\right)\bigr\}(\lambda)}{\omega},
\end{equation}
with the Hurst exponent $H\in(\alpha_{\mathcal{G}},\beta_{\mathcal{G}})$, and subsequently~the \ac{LDS} of $\mathcal{H}_{avg}(\gammabar)$ as
\begin{equation}
\label{Eq:LDSForHavg}
\mathcaltilde{H}_{avg}\left(\omega,\gammabar\right)
	\trigeq\FourierTransform{\lambda}{\mathfrak{L}^{-1}_{H,\gammabar}\bigl\{\mathcal{H}_{avg}\left(\gammabar\right)\bigr\}(\lambda)}{\omega},
\end{equation}
with the Hurst exponent $H\!\in\!(\alpha_{\mathcal{H}},\beta_{\mathcal{H}})$. Therefore, the ratio of $\mathcaltilde{H}_{avg}(\omega,\gammabar)$ and $\mathcaltilde{G}_{avg}(\omega,\gammabar)$ does essentially exist for  $H\in(\max(\alpha_{\mathcal{G}},\alpha_{\mathcal{H}}),\min(\beta_{\mathcal{G}},\beta_{\mathcal{H}}))$ and is readily obtained by applying \eqref{Eq:ScaleInvariantSpectrum} on \eqref{Eq:DilationBasedRelationshipAmongAPMs} and then using \eqref{Eq:ScaleInvariantSpectrumDilation}, that is
\begin{equation}\label{Eq:LimpertiDilationSpectrumRatio}
\frac{\raisebox{0.382ex}{$\mathcaltilde{H}_{avg}^{*}\left(\omega,\gammabar\right)$}}
	{\mathcaltilde{G}_{avg}^{*}\left(\omega,\gammabar\right)}
			\trigeq\sum_{n=1}^{N}z_{n}\,\lambda_{n}^{H-\imaginary\,\omega},
\end{equation}
where the superscript $*$ denotes the complex conjugation. It is worth emphasizing that \eqref{Eq:LimpertiDilationSpectrumRatio} is noteworthily~independent~of~$\gammabar$. Further, the dilations are positive, i.e., $\lambda_{n}\in\mathbb{R}_{+}$ for all $n\in\{1,\allowbreak{2},\ldots,N\}$. The right side of \eqref{Eq:LimpertiDilationSpectrumRatio} can be therefore observed as a signal filter whose parameters $\{z_{n}\}_{1}^{N}$ and $\{\lambda_{n}\}_{1}^{N}$ are to be determined independently from $\gammabar$. With that context, note that the \ac{LT} of Dirac's delta function is given by
\begin{equation}\label{Eq:DiracLampertiTransformation}
\mathfrak{L}^{-1}_{H,u}\!\bigl\{\DiracDelta{u-\lambda_{n}}\bigr\}(\lambda)={e}^{H\lambda}\,\delta\bigl({e}^{-\lambda}\,u-\lambda_{n}\bigr),
\end{equation}
for any Hurst exponent  $H\in\mathbb{R}$ and any dilation $\lambda_{n}\in\mathbb{R}_{+}$, where $\delta\left(\cdot\right)$ denotes Dirac's delta function\cite[Eq.(1.8.1)]{BibZwillingerBook}. By using this result, the \ac{LDS} of Dirac's delta function, i.e., the \ac{FT} of \eqref{Eq:DiracLampertiTransformation} is obtained as
\begin{subequations}\label{Eq:DiracLimpertiDilationSpectrum}
\setlength\arraycolsep{1.4pt}
\begin{eqnarray}
\label{Eq:DiracLimpertiDilationSpectrumA}
    \Delta_{n}\bigl(\omega,H\bigr)&\trigeq&
        \FourierTransform{\lambda}
            {\mathfrak{L}^{-1}_{H,u}\!\left\{\DiracDelta{u-\lambda_{n}}\right\}\!(\lambda)}
                {\omega},\\
\label{Eq:DiracLimpertiDilationSpectrumB}
    &=&\int_{-\infty}^{+\infty}
		{e}^{\imaginary\,\omega\lambda}\,
			\mathfrak{L}^{-1}_{H,u}\!\left\{\DiracDelta{u-\lambda_{n}}\right\}\!(\lambda)\,
				d\lambda,{~~~}\\
\label{Eq:DiracLimpertiDilationSpectrumC}
    &=&{\lambda_{n}}^{-(H+1)-\imaginary\omega}{u}^{H+\imaginary\omega}.
\end{eqnarray}
\end{subequations}
As a consequence of \eqref{Eq:DiracLimpertiDilationSpectrum}, and examining the right part of \eqref{Eq:LimpertiDilationSpectrumRatio}, we can now write 
$\lambda_{n}^{H-\imaginary\omega}={\Delta_{n}\bigl(\omega,-(H+1)\bigr)}/\allowbreak{u^{-(H+1)+\imaginary\omega}}$ and hence reduce \eqref{Eq:LimpertiDilationSpectrumRatio} to 
\begin{equation}
    \label{Eq:LimpertiDilationSpectrumForAuxiliaryFunction} 
    \sum_{n=1}^{N}z_{n}\,\lambda_{n}^{H-\imaginary\,\omega}=
        \frac{\FourierTransform{\lambda}
                {\mathfrak{L}^{-1}_{-(H+1),u}\bigl\{\mathcal{Z}_{N}(u)\bigr\}(\lambda)}
                    {\omega}}
            {{u}^{-(H+1)+\imaginary\,\omega}},
\end{equation}
where $\mathcal{Z}_{N}(u)$ is an auxiliary function deduced from \eqref{Eq:DilationBasedRelationshipAmongAPMs} as
\begin{equation}\label{Eq:DiscreteAuxiliaryFunction}
\mathcal{Z}_{N}(u)=\sum_{n=1}^{N}z_{n}\,\DiracDelta{u-\lambda_{n}},
\end{equation}
where we need to determine the weights $z_{1},z_{2},\ldots,z_{N}$ and the dilations $\lambda_{1},\lambda_{2},\dots,\lambda_{N}$. Within that context, we reasonably deduce that $\mathcal{Z}_{N}(u)$ is a discretized version of the continuous auxiliary function $\mathcal{Z}(u)$ in such a way that, for all $n\in\{1,\allowbreak{2},\ldots,N\}$, we implicitly consider $z_{n}$ as a sample taken from $\mathcal{Z}(u)$ at the dilation $\lambda_{n}$ and therein choose the total number of samples $N$ as large as possible according to the required precision, i.e., $\mathcal{Z}(u)=\lim_{N\rightarrow\infty}\mathcal{Z}_{N}(u)$. Accordingly, we can establish a relationship between two \acp{APM} as described in the following theorem.   
\ifCLASSOPTIONtwocolumn
\pagebreak[4]
\fi

\begin{theorem}[Relationship between two \acp{APM}]
\label{Theorem:RelationshipBetweenAPMs}
A relationship between two \acp{APM} $\mathcal{G}_{avg}\left(\gammabar\right)$ and $\mathcal{H}_{avg}\left(\gammabar\right)$ is given by
\begin{equation}\label{Eq:RelationshipBetweenAPMs}
\mathcal{H}_{avg}\left(\gammabar\right)=
    \int_{0}^{\infty}\mathcal{Z}\left(u\right)\,
        \mathcal{G}_{avg}\left(u\,\gammabar\right)\,du,
\end{equation}
where $\mathcal{Z}(u)$ is an auxiliary function defined by
\begin{equation}\label{Eq:AuxiliaryFunctionForRelationshipBetweenAPMs}
\!\!\mathcal{Z}(u)=\mathfrak{L}_{-(H+1),\lambda}
	\Biggl\{
		\mathfrak{F}_{\omega}^{-{1}}
		\biggl\{
			\frac{\raisebox{0.382ex}{${u}^{\imaginary\,\omega}\,\mathcaltilde{H}_{avg}^{*}\left(\omega,\gammabar\right)$}}
				{{u}^{H+1}\,\mathcaltilde{G}_{avg}^{*}\left(\omega,\gammabar\right)}\!
		\biggr\}\{\lambda\}\!
	\Biggr\}(u),
\end{equation}
whose existence is verified by choosing the Hurst exponent $H$ such that the \acp{FT}, given in both \eqref{Eq:LDSForGavg} and \eqref{Eq:LDSForHavg}, are convergent.
\end{theorem}

\begin{proof}
Referring to \eqref{Eq:DilationBasedRelationshipAmongAPMs}, which corresponds the relationship we want to achieve, we attempt to obtain the \ac{APM} $\mathcal{H}_{avg}(\gammabar)$ for a certain average \ac{SNR} $\gammabar$ by means of the measurement set $\mathcal{S}_N$, given in \eqref{Eq:MeasurementSet}, that is obtained by experimental measurement or theoretical calculation of the other \ac{APM} $\mathcal{G}_{avg}(\gammabar)$. Using \cite[Eq.(1.8.1/1)]{BibZwillingerBook}, we can regulate and re-express the dilated performance $\mathcal{G}_{avg}(\lambda_{n}\gammabar)$ as 
$\mathcal{G}_{avg}(\lambda_{n}\gammabar)=\int_{0}^{\infty}\DiracDelta{u-\lambda_{n}}\allowbreak\,\mathcal{G}_{avg}(u\gammabar)\,du$, and therefrom we readily rewrite \eqref{Eq:DilationBasedRelationshipAmongAPMs} using the discrete
auxiliary function $\mathcal{Z}_{N}(u)$ that is given in \eqref{Eq:DiscreteAuxiliaryFunction} as 
\begin{subequations}
\setlength\arraycolsep{1.4pt}
\begin{eqnarray}
    \mathcal{H}_{avg}(\gammabar)
    &=&\lim_{N\rightarrow\infty}\sum_{n=1}^{N}z_{n}\int_{0}^{\infty}
        \DiracDelta{u-\lambda_{n}}\,\mathcal{G}_{avg}(u\gammabar)\,du,{~~~~~}\\
    &=&\lim_{N\rightarrow\infty}\int_{0}^{\infty}
        \sum_{n=1}^{N}z_{n}\DiracDelta{u-\lambda_{n}}\,
            \mathcal{G}_{avg}(u\gammabar)\,du,{~~~~~}\\
    &=&\lim_{N\rightarrow\infty}\int_{0}^{\infty}
        \mathcal{Z}_{N}(u)\,\mathcal{G}_{avg}(u\gammabar)\,du,\\
    &=&\int_{0}^{\infty}
        \mathcal{Z}(u)\,\mathcal{G}_{avg}(u\gammabar)\,du,
\end{eqnarray}
\end{subequations}
which proves the relationship given in \eqref{Eq:RelationshipBetweenAPMs} and completes the first part of the proof. In the the second part, we will find the auxiliary function $\mathcal{Z}(u)$. First, with the aid of the result that we readily obtain substituting the left-hand side of \eqref{Eq:LimpertiDilationSpectrumRatio} into
\eqref{Eq:LimpertiDilationSpectrumForAuxiliaryFunction}, we simplify the problem of finding
the weights $z_{1},z_{2},\ldots,z_{N}$ and the dilations $\lambda_{1},\lambda_{2},\dots,\lambda_{N}$ to achieving the \ac{LDS}
of the continuous auxiliary function $\mathcal{Z}(u)$. In more details, 
using $\mathcal{Z}(u)=\lim_{N\rightarrow\infty}\mathcal{Z}_{N}(u)$ and referring both to  
\eqref{Eq:LimpertiDilationSpectrumRatio} and
\eqref{Eq:LimpertiDilationSpectrumForAuxiliaryFunction}, we have 
\begin{equation}
\frac{\raisebox{0.382ex}{$\mathcaltilde{H}_{avg}^{*}\left(\omega,\gammabar\right)$}}
	{\mathcaltilde{G}_{avg}^{*}\left(\omega,\gammabar\right)}
			\trigeq
			\frac{\FourierTransform{\lambda}
                {\mathfrak{L}^{-1}_{-(H+1),u}\bigl\{\mathcal{Z}(u)\bigr\}(\lambda)}
                    {\omega}}
            {{u}^{-(H+1)+\imaginary\,\omega}},
\end{equation}
for any Hurst exponent $H$ such that both the \acp{FT} of \eqref{Eq:LDSForGavg} and that of \eqref{Eq:LDSForHavg} are convergent. After performing some algebraic manipulations, we obtain the \ac{LDS} of $\mathcal{Z}(u)$ as follows
\begin{equation}\label{Eq:LampertiDilationSpectrumOfAuxiliaryFunctionAsRatioOfPerformanceLampertiTransforms}
\FourierTransform{\lambda}{\!\mathfrak{L}^{-{1}}_{-(H+1),u}\bigl\{\mathcal{Z}(u)\bigr\}(\lambda)\!}{\omega}
	=\frac{\raisebox{0.382ex}{${u}^{\imaginary\,\omega}\,\mathcaltilde{H}_{avg}^{*}\!\left(\omega,\gammabar\right)$}}
			{{u}^{H+1}\,\mathcaltilde{G}_{avg}^{*}\!\left(\omega,\gammabar\right)}.
\end{equation}
where applying the \ac{IFT} and then exercising the \ac{LT}, we readily deduce the continuous auxiliary function
$\mathcal{Z}(u)$ as \eqref{Eq:AuxiliaryFunctionForRelationshipBetweenAPMs}, which completes the second part of the proof and thus completes the proof of \theoremref{Theorem:RelationshipBetweenAPMs}.
\end{proof}

The relationship between $\mathcal{H}_{avg}(\gammabar)$ and $\mathcal{G}_{avg}(\gammabar)$ enables us to investigate $\mathcal{H}_{avg}(\gammabar)$ approximately using \ac{GCQ} formula\cite[Eq.~(11a)~and~(11b)]{BibYilmazAlouiniTCOM2012}, that is 
\begin{subequations}
\label{Eq:GCQBasedRelationBetweenAveragedPerformanceMetrics}
\setlength\arraycolsep{1.4pt}
\begin{eqnarray}
    \label{Eq:GCQBasedRelationBetweenAveragedPerformanceMetricsA}
    \mathcal{H}_{avg}(\gammabar)&\approx&
        \sum_{n=1}^{N}w_{n}\,
            \mathcal{Z}(\lambda_{n})\,
            \mathcal{G}_{avg}(\lambda_{n}\gammabar),\\
    \label{Eq:GCQBasedRelationBetweenAveragedPerformanceMetricsB}
    &=&
        \sum_{n=1}^{N}z_{n}\,
            \mathcal{G}_{avg}(\lambda_{n}\gammabar),
\end{eqnarray}
\end{subequations}
which is called the \ac{QBP} technique, where the dilation is $\lambda_{n}=\tan(\frac{\pi}{4}\cos(\frac{2n-1}{2N}\pi)+\frac{\pi}{4})$ and the weight is $w_{n}=\frac{\pi^2}{4N}\sin(\frac{2n-1}{2N}\pi)\sec^{2}(\frac{\pi}{4}\cos(\frac{2n-1}{2N}\pi)+\frac{\pi}{4})$. It is worth noting that $N$ has to be chosen as large as possible for an accurate approximation. The existence of such a relationship between $\mathcal{H}_{avg}(\gammabar)$ and $\mathcal{G}_{avg}(\gammabar)$ depends on the existence of their \ac{LDS} spectrums as explained in the following theorem.

\begin{theorem}[Existence of a Relationship Between~Two~\acp{APM}]\label{Theorem:RelationshipExistenceBetweenAPMs}
Assume that the \ac{LDS} of $\mathcal{G}_{avg}(\gammabar)$ exist for any Hurst exponent
$H\!\in\!(\alpha_{\mathcal{G}},\beta_{\mathcal{G}})$, and $\mathcal{H}_{avg}(\gammabar)$ for any Hurst exponent
$H\!\in\!(\alpha_{\mathcal{H}},\beta_{\mathcal{H}})$. The relationship given by \theoremref{Theorem:RelationshipBetweenAPMs}, certainly exists if and only if $(\alpha_{\mathcal{H}},\beta_{\mathcal{H}})\cap(\alpha_{\mathcal{G}},\beta_{\mathcal{G}})\neq\emptyset$.
\end{theorem}

\begin{proof}
The proof is evident referring to the existence of
\eqref{Eq:LampertiDilationSpectrumOfAuxiliaryFunctionAsRatioOfPerformanceLampertiTransforms} based on \corollaryref{Theorem:LDSExistence}.
\end{proof}

\subsection{Relation to Mellin's Convolution}
\label{Section:RelationshipsAmongAPMs:MellinConvolution}
The relationship~between~two \acp{APM}, which~is~given~in~\theoremref{Theorem:RelationshipBetweenAPMs} and whose existence is proven in \theoremref{Theorem:RelationshipExistenceBetweenAPMs} using the \ac{LDS} spectrums of the \acp{APM}, is a multiplicative kind of integral transform known as the Mellin's convolution\cite{BibOberhettingerBook,BibPoularikasBook,BibKilbasSaigoBook,BibMathaiSaxenaHauboldBook} in the literature. It is worth mentioning that~the~Mellin's~convolution is an extremely powerful technique, which is readily understood by non-specialists in integral transforms and special functions, for the exact evaluation of definite integrals, and \emph{it can often result in closed-form expressions~for~the~most general case, using higher transcendental functions such as hypergeometric, Meijer'G and Fox's H functions}\cite{BibKilbasSaigoBook}. Thus, we notice that the pairs of \ac{MT} and \ac{IMT}, which are given largely in several satisfactory tables\cite{BibOberhettingerBook}, yield not only fast numerical computations and also tractable closed-form results. As such, the auxiliary function $\mathcal{Z}(u)$ can be readily obtained in terms of \acp{MT} of the \acp{PM} with respect to the average \ac{SNR}. The \ac{MT} of $\mathcal{H}_{avg}(\gammabar)$ is written as $\MellinTransform{\gammabar}{\mathcal{H}_{avg}(\gammabar)}{s}=\int_{0}^{\infty}\gammabar^{s-1}\mathcal{H}_{avg}(\gammabar)d\gammabar$, where $\MellinTransform{\cdot}{\cdot}{\cdot}$ denotes the \ac{MT}\footnote{\label{Footnote:MellinTransform} Let $\psi\colon\mathbb{R}^{+}\to\mathbb{R}$ be a real-valued and monotonic function~locally~integrable and differentiable. The \acf{MT} of this function is defined as
\begin{equation}
	\Psi(s)=
		\MellinTransform{\gamma}{\psi(\gamma)}{s}
			=\int_{0}^{\infty}\gamma^{s-1}\,\psi(\gamma)\,d\gamma,
				\tag{F.\thefootnote.1}
\end{equation}
for $s\in\mathcal{R}_{\mathcal{OC}}$, where $\mathcal{R}_{\mathcal{OC}}$ is the region of convergence (ROC) and defined as $\mathcal{R}_{\mathcal{OC}}\bigl\{\MellinTransform{\gamma}{\psi(\gamma)}{s}\bigr\}=\bigl\{s\in\mathbb{C}\,\bigl|\,\int_{0}^{\infty}\abs{\gamma^{s-1}\,\phi(\gamma)}\,d\gamma<\infty\bigr.\bigr\}$. Further, the \acf{IMT} is defined as
\begin{equation}
	\psi(\gamma)=
		\InvMellinTransform{s}{\Psi(s)}{\gamma}
			=\frac{1}{2\pi\imaginary}\int_{\mathcal{C}}\Psi(s)\,\gamma^{-s}\,ds,
				\tag{F.\thefootnote.2}
\end{equation}
where the contour integration $\mathcal{C}\in\mathcal{R}_{\mathcal{OC}}\{\MellinTransform{\gamma}{\psi(\gamma)}{s}\}$ is chosen to be counterclockwise in order to ensure the convergence. For more information, the readers and researchers are referred to \cite{BibOberhettingerBook,BibPoularikasBook,BibKilbasSaigoBook,BibMathaiSaxenaHauboldBook}.}. Therein, replacing~\eqref{Eq:RelationshipBetweenAPMs}~and~using \cite[Eq. (1.2)]{BibOberhettingerBook} after changing the order of integrals, we write   
\begin{equation}\nonumber
\!\MellinTransform{\gammabar}{\mathcal{H}_{avg}(\gammabar)}{s}=
    \int_{0}^{\infty}\mathcal{Z}(u)\,
	    \biggl\{\int_{0}^{\infty}\gammabar^{s-1}
	        \mathcal{G}_{avg}(u\gammabar){d\gammabar}\biggr\}{du},
\end{equation}
where making use of
$\MellinTransform{\gammabar}{\mathcal{G}_{avg}(\gammabar)}{s}=\int_{0}^{\infty}\gammabar^{s-1}\mathcal{G}_{avg}(\gammabar)d\gammabar$ and $\MellinTransform{u}{\mathcal{Z}(u)}{s}=\int_{0}^{\infty}{u}^{s-1}\mathcal{Z}(u)du$ yields 
\begin{equation}
\MellinTransform{u}{\mathcal{Z}(u)}{s}=
	\frac{\MellinTransform{\gammabar}{\mathcal{H}_{avg}(\gammabar)}{1-s}}
		{\MellinTransform{\gammabar}{\mathcal{G}_{avg}(\gammabar)}{1-s}},
\end{equation}
where applying the \ac{IMT} results in the auxiliary function $\mathcal{Z}(u)$, that is given by 
\begin{equation}\label{Eq:AuxiliaryFunctionUsingMellinTransform}
\mathcal{Z}\left(u\right)=\InvMellinTransform{s}{\frac{\MellinTransform{\gammabar}{\mathcal{H}_{avg}(\gammabar)}{1-s}}
	{\MellinTransform{\gammabar}{\mathcal{G}_{avg}(\gammabar)}{1-s}}}{u},
\end{equation}
where $\InvMellinTransform{\cdot}{\cdot}{\cdot}$ denotes the \ac{IMT}.\footnoteref{Footnote:MellinTransform}

\subsection{Empirical (Experimental) Usage}
\label{Section:RelationshipsAmongAPMs:EmpiricalAPMEstimation}
In general, as regards to empirical performance prediction, we want to calculate the \ac{APM} $\mathcal{H}_{avg}(\gammabar)$ using the other \ac{APM} $\mathcal{G}_{avg}(\gammabar)$, without having the knowledge of \ac{SNR} distribution and the broadest \ac{SNR} settings, where we assume that $\mathcal{H}_{avg}(\gammabar)$ is either difficult or impossible to measure empirically but $\mathcal{G}_{avg}(\gammabar)$ is easy~to~measure~empirically. Let us assume that the \ac{APM} $\mathcal{G}_{avg}(\gammabar)$ is experimentally measured for
\begin{equation}\label{Eq:AverageSNRs}
    \gammabar_{1}\leq\gammabar_{2}\leq\ldots\leq\gammabar_{N-1}\leq\gammabar_{N},
\end{equation}
where we arbitrarily choose $\gammabar_{n}$ as $-\gammabar_\text{dB}<10\log_{10}(\gammabar_{n})<\gammabar_\text{dB}$ for all $n\in\{1,2,\ldots,N\}$ and hence therefrom~obtain~the~measurement~set $\mathcal{S}_N=\{(\gammabar_{1},\mathcal{G}_{1}),(\gammabar_{2},\mathcal{G}_{2}),\ldots,(\gammabar_{N},\mathcal{G}_{N})\}$, where $\mathcal{G}_{n}=\mathcal{G}_{avg}(\gammabar_n)$ for all $n\in\{1,2,\ldots,N\}$. There are many interpolation techniques\cite{BibDavisBook1975,BibPressTeukolskyNumericalRecipes1992Book,BibWong2010Book} available in the literature, each one of which can be readily applied on the measurement set $\mathcal{S}_N$ to approximately reproduce $\mathcal{G}_{avg}(\gammabar)$ for any $\gammabar\in\mathbb{R}_{+}$. The Lagrange's interpolation is one of them implemented as a built-in function in standard mathematical software packages such as \Mathematica, \Maple and \Matlab. The Lagrange's interpolation of $\mathcal{G}_{avg}(\gammabar)$ is denoted by $\mathcal{G}_{int}\left(\gammabar|\mathcal{S}_N\right)$, and on the measurement set $\mathcal{S}_N$, written as \cite[Eq. (2.5.3)]{BibDavisBook1975}
\begin{equation}\label{Eq:AveragedPerformanceApproximatedWithInterpolation}
    \mathcal{G}_{int}\left(\gammabar\,|\,\mathcal{S}_N\right)
        \trigeq\sum_{n=1}^{N}
            \mathcal{G}_{n}
            \prod_{\substack{k=1,\,{k}\neq{n}}}^{N}
                \frac{\gammabar-\gammabar_k}{\gammabar_{n}-\gammabar_k},
\end{equation}
which is a polynomial of degree $N\mathord{-}1$ coinciding with $\mathcal{G}_{avg}(\gammabar)$ at  $\gammabar_1,\gammabar_2,\ldots,\gammabar_N$. 

To find the error committed in the Lagrange's interpolation, we can write $\mathcal{G}_{avg}(\gammabar)=\mathcal{G}_{int}\left(\gammabar\,|\,\mathcal{S}_N\right)+R_{N}(\gammabar)$ for the interval $\mathbb{G}=[\gammabar_1,\gammabar_N]$, where $R_{N}(\gammabar)$ is the interpolation error that we can obtain exploiting Cauchy remainder theorem \cite{BibDavisBook1975}, that is 
\begin{equation}\label{Eq:APMLagrangeInterpolationError}
    R_{N}(\gammabar)=\frac{1}{N!}
        \frac{\partial^N}{d{g}^N}\mathcal{G}_{avg}(g)
            \prod_{k=1}^{N}(\gammabar-\gammabar_k),
\end{equation}
where $g$ is an intermediate point in $\mathbb{G}$, where we observe that, due to the existence of $N!$ in the denominator, the absolute error $\abs{R_{N}(\gammabar)}$ decreases very quickly while the measurement number $N$ increases. In addition, in accordance with \ac{SNR}-incremental Gaussian channel\cite{BibDongningGuoShamaiVerduTIT2005}, we choose $\gammabar_1,\gammabar_2,\ldots,\gammabar_N$ as exponential spaced points, i.e., $10\log(\gammabar_n)=\bigl(2\frac{n-1}{N-1}-1\bigr)\gammabar_{dB}$. Then, replacing \eqref{Eq:AveragedPerformanceApproximatedWithInterpolation} in \eqref{Eq:RelationshipBetweenAPMs}, we are able to estimate $\mathcal{H}_{avg}(\gammabar)$ from the empirical measurements of $\mathcal{G}_{avg}(\gammabar)$, which we call the \ac{IBP} technique calculating any \ac{APM} using the other empirically measured \acp{APM}, that is
\begin{equation}
    \mathcal{H}_{avg}(\gammabar)=
            \int_{0}^{\infty}
                \mathcal{Z}\left(u\right)\,
                \mathcal{G}_{int}\left(u\,\gammabar\,|\,\mathcal{S}_{N}\right)du+
                E_{N}(\gammabar),
\end{equation}
where the term $E_{N}(\gammabar)$ denotes the estimation error. Substituting $\mathcal{G}_{avg}(\gammabar)=\mathcal{G}_{int}\left(\gammabar\,|\,\mathcal{S}_N\right)+R_{N}(\gammabar)$ in \eqref{Eq:RelationshipBetweenAPMs} and therein using \eqref{Eq:APMLagrangeInterpolationError}, the absolute error term $\abs{E_{N}(\gammabar)}$ is bounded as 
\begin{equation}
    \abs{E_{N}(\gammabar)}\leq
        \frac{1}{N!}
            \abs{Z_{N}(\gammabar)}
            \sup_{g\,\in\,\mathbb{G}}
                \abs{\frac{\partial^N}{d{g}^N}\mathcal{G}_{avg}(g)},
\end{equation}
where $Z_{N}(\gammabar)=\int_{0}^{\infty}\mathcal{Z}\left(u\right)\prod_{k=1}^{N}(u\gammabar-\gammabar_k)\,du$.

\section{\ac{ACC} Analysis Using \ac{ABER}}
\label{Section:RelationshipBetweenACCandABEP}
For a limited-bandwidth complex \ac{AWGN} channel,  the most celebrated result in the literature is the channel capacity $C(\gamma)$, which is given as 
\cite{BibShannonBSTJ1948,BibShannon1949,BibShannonWeaverUIP1949}
\begin{equation}\label{Eq:ChannelCapacityInAWGNChannels}
\mathcal{C}(\gamma)=\log\left(1+\gamma\right)~~\text{nats\,/\,s\,/\,Hz},
\end{equation}
where $\gamma$ is the instantaneous \ac{SNR}, and where $\log(\cdot)$ is the natural logarithm. 
It confirms that the maximum information throughput is achievable with asymptotically small error probability, such that
a reliable transmission is possible for the information throughput $\mathcal{R}\leq\mathcal{C}(\gamma)$. In the literature of channel capacity, researches are commonly based on \eqref{Eq:ChannelCapacityInAWGNChannels} but
explicitly achieved by extending its definition to the \ac{ACC} performance
\cite{BibLeeTVT1990}, especially when the \ac{CSI} is
known at the receiver \cite{BibEricsonTIT1970,BibOzarowShamaiWynerTVT1994}. Referring to the discussion in \secref{Section:RelationshipsAmongAPMs}, the \ac{ACC} is given
by $\mathcal{C}_{avg}(\gammabar)\trigeq\mathbb{E}\bigl[\mathcal{C}(\gamma(\boldsymbol{\psi}))\bigr]=\mathbb{E}\left[\log\left(1+\gamma(\boldsymbol{\psi})\right)\right]=\int\int\ldots\int
	\log\bigl(1+{\gamma\left(\boldsymbol{\psi}\right)}\bigr)\allowbreak
		{p}_{\boldsymbol{\Psi}}(\boldsymbol{\psi})\,{d}\boldsymbol{\psi}$
\cite{BibGoldsmithVaraiyaTIT1997,BibAlouiniGoldsmithTVT1999,BibAlouiniGoldsmithWPC2000}\footnote{Adaptive transmission schemes utilize the acquisition of \ac{CSI} at both the receiver and transmitter to change the \ac{SNR} distribution as $\gamma=D(\gammatilde)\,\gammatilde$, where $\gammatilde$ denotes the \ac{SNR} distribution at the receiver before the power adaptation function $D(\gamma)$ is applied (i.e., see \cite{BibGoldsmithVaraiyaTIT1997,BibAlouiniGoldsmithTVT1999,BibAlouiniGoldsmithWPC2000,BibViswanathanTIT1999,BibJayaweeraPoorTIT2003,BibYuanZhangTepedelenliogluTIT2012,BibGoldsmithBook,BibAlouiniBook} for more details), namely supporting the average power constraint, namely,
\vspace{-1mm}
\begin{equation}
\mathbb{E}\left[\gammatilde\right]=
	\mathbb{E}\left[D(\gammatilde)\right]=
		\int_{0}^{\infty}D(\gamma)\,p_{\gammatilde}(\gamma)\,d\gamma,
\tag{F.\thefootnote.1}
\vspace{-1mm}
\end{equation}
where $p_{\gammatilde}(\gamma)=\mathbb{E}[\delta(\gamma-\gammatilde)]$ denotes the \ac{PDF} of $\gammatilde$. In accordance, $\mathbb{E}[\log(1+D(\gammatilde)\,\gammatilde)]$ can be rewritten as $\mathcal{C}_{avg}(\gammabar)=\mathbb{E}[\log(1+\gamma)]$ where
$\gamma=D(\gammatilde)\,\gammatilde$.}, where $\boldsymbol{\Psi}$ and $\gammabar$ denote the \ac{SNR} settings and the average \ac{SNR}, respectively, explained in the first lines of \secref{Section:RelationshipsAmongAPMs}. The other most celebrated result in the literature is the \ac{BER} \cite[and~references~therein]{BibAlouiniBook},\cite{BibSimonAlouiniProcIEEE1998,BibAlouiniGoldsmithTCOM1999}, compactly denoted as $\mathcal{E}(\gamma)$ and modeled as a random distribution between zero and one-half for all modulation schemes, i.e., $0<\mathcal{E}(\gamma)<{1}/{2}$. Accordingly, there exists a reliability metric $\mathcal{Q}(\gamma)$, which we term the \ac{CR}, specifically defined in terms of the \ac{BER}, that is
\begin{equation}\label{Eq:ChannelReliabilityForBinaryModulations}
	\mathcal{Q}\left(\gamma\right)=
	    {1}-{2}\,\mathcal{E}\left({\gamma}\right),
\end{equation}
which possesses knowledge about how reliably the information are transferred through the channel and therefore has a close similarity to the channel capacity. Further, it has such a distributional behavior that ${0}<\mathcal{Q}(\gamma)<{1}$, where the signaling channel is fully dissipated (i.e., the transferred entropy through the signalling channel becomes zero) when $\mathcal{Q}\left(\gamma_{end}\right)=0$ and is error-free when $\mathcal{Q}\left(\gamma_{end}\right)=1$. The \ac{ACR}, denoted by $\mathcal{Q}_{avg}\left(\gammabar\right)=\Expected{\mathcal{Q}\left(\gamma\right)}$, can be calculated in an averaging sense, namely 
\begin{equation}\label{Eq:AveragedChannelReliability}
	\mathcal{Q}_{avg}\left(\gammabar\right)=1-2\,\mathcal{E}_{avg}\left(\gammabar\right),
\end{equation}
where $\mathcal{E}_{avg}(\gammabar)$ denotes the \ac{ABER} and is defined by $\mathcal{E}_{avg}(\gammabar)=\mathbb{E}[\mathcal{E}(\gamma)]=\int\int\ldots\int\mathcal{E}({\gamma(\boldsymbol{\psi})})p_{\boldsymbol{\Psi}}(\boldsymbol{\psi})\,{d}\boldsymbol{\psi}$. It is within that~context~that either measuring the \ac{ABER} experimentally or deriving it mathematically for different average \acp{SNR}\cite[and~references~therein]{BibAlouiniBook} is seemingly trivial and quite straightforward
compared to that of the \ac{ACC} performance. In what follows, a relationship between the \ac{ACC} and the \ac{ABER}
is given for specific modulation schemes. 

\subsection{Binary Modulation Schemes}
\label{Section:RelationshipBetweenACCandABEP:BinaryModulationSchemes}
The Wojnar's unified \ac{BER} for binary modulation schemes is given by
$\mathcal{E}\left(\gamma\right)=({1}/{2})\,{\Gamma\left(b,a\,\gamma\right)}/{\Gamma\left(b\right)}$
\cite[Eq. (13)]{BibWojnarTCOM1986} and \cite[Eq. (8.100)]{BibAlouiniBook}, where the value of $a$ depends on the type of modulation scheme (${1}/{2}$ for orthogonal FSK and $1$ for antipodal PSK), the value of $b$
depends on the type of detection technique (${1}/{2}$ for coherent and $1$ for non-coherent), and $\Gamma\left(\cdot\right)$ and
$\Gamma\left(\cdot,\cdot\right)$ are the Gamma function \cite[Eq.~(6.1.1)]{BibAbramowitzStegunBook} and the complementary
incomplete Gamma function \cite[Eq.~(6.5.3)]{BibAbramowitzStegunBook}, respectively.
Properly in connection with \eqref{Eq:AveragedChannelReliability}, applying
\cite[Eqs.~(8.4.16/1)~and~(8.4.16/1)]{BibPrudnikovBookVol3} to $\mathcal{E}\left(\gamma\right)$, we write
\begin{equation}
    \label{Eq:EqCertainityForBinaryModulations}
    \mathcal{Q}\left(\gamma\right)=
        1-\frac{\Gamma\left(b,a\,\gamma\right)}{\Gamma\left(b\right)}=
            \frac{\widehat\Gamma\left(b,a\,\gamma\right)}{\Gamma\left(b\right)},
\end{equation}
where $\widehat\Gamma\left(\cdot,\cdot\right)$ is the lower incomplete Gamma function\cite[Eq.~(8.350/1)]{BibGradshteynRyzhikBook} such that $\Gamma\left(b,a\gamma\right)+\widehat\Gamma\left(b,a\gamma\right)=\Gamma\left(b\right)$. 

\begin{theorem}[\ac{ACC} analysis using \ac{ABER} of binary modulation schemes]
\label{Theorem:RelationshipBetweenACCAndABEPForBinaryModulationSchemes}
Let a wireless communication system use a binary modulation for signaling in fading environments. Then, its \ac{ACC} $\mathcal{C}_{avg}(\gammabar)$ is obtained from using its \ac{ABER} $\mathcal{E}_{avg}(\gammabar)$ as 
\begin{subequations}
\label{Eq:RelationshipBetweenACCAndABEPForBinaryModulationSchemes}
\vspace{-4mm}
\setlength\arraycolsep{1.4pt}
\begin{eqnarray}
    \label{Eq:RelationshipBetweenACCAndABEPForBinaryModulationSchemesA}
    \mathcal{C}_{avg}\left(\gammabar\right)
        &=&\int_{0}^{\infty}\!\!\mathcal{Z}_{a,b}\left(u\right)\mathcal{Q}_{avg}\left(u\,\gammabar\right)du,\\
    \label{Eq:RelationshipBetweenACCAndABEPForBinaryModulationSchemesB}
    &=&\int_{0}^{\infty}\!\!\mathcal{Z}_{a,b}\left(u\right)\Bigl\{{1}-2\,\mathcal{E}_{avg}\left(u\,\gammabar\right)\Bigr\}du,    
\end{eqnarray}
\end{subequations}
where the auxiliary function $\mathcal{Z}_{a,b}(u)$ is defined by
\begin{equation}\label{Eq:AuxiliaryFunctionForACCAndWojnarABEP}
    \mathcal{Z}_{a,b}\left(u\right)=\frac{1}{u}\,\Hypergeom{1\!}{1}{1}{b}{-a\,u},
\end{equation}
where $a$ and $b$ are the modulation specific parameters explained above, and $\Hypergeom{1}{1}{\cdot}{\cdot}{\cdot}$ denotes Kummer's confluent hypergeometric function \emph{\cite[Eq.~(07.20.02.0001.01)]{BibWolfram2010Book}}.
\end{theorem}

\begin{proof}
Due to the monotonic increasing nature of \ac{ACC} and \ac{ACR} (i.e. since $\mathcal{Q}_{avg}(\gammabar)\leq\mathcal{Q}_{avg}(\gammabar+\Delta\gammabar)$ and $\mathcal{C}_{avg}(\gammabar)\leq\mathcal{C}_{avg}(\gammabar+\Delta\gammabar)$ for $\Delta\gammabar\in\mathbb{R}_{+}$), both the \ac{ACC} and the \ac{ACR} together preserve the scaling order according to  \theoremref{Theorem:RelationshipExistenceBetweenAPMs}, and therefore their \acp{LDS} surely exist for a mutually common Hurst's exponent. We derive the \ac{LDS} of the \ac{ACC} using \cite[Eq. (4.293/3) and (4.293/10)]{BibGradshteynRyzhikBook}, that is
\begin{subequations}
    \label{Eq:ACCDilationSpectrum}
    \setlength\arraycolsep{1.4pt}
    \begin{eqnarray}
    \label{Eq:ACCDilationSpectrumA}
    \mathcaltilde{C}_{avg}\left(\omega,\gammabar\right)&=&
            \mathfrak{F}_{\lambda}\!\bigl\{\mathfrak{L}^{-1}_{H,\gammabar}
                \bigl\{\mathcal{C}_{avg}\left(\gammabar\right)\bigr\}(\lambda)\bigl\}\left({\omega}\right),\\
    \label{Eq:ACCDilationSpectrumB}
        &=&\mu_{\gamma}(H+\imaginary\,\omega;\gammabar)\!
		    \int_{0}^{\infty}\!\frac{\lambda^{-\imaginary\,\omega}}{\lambda^{H+1}}
		        \log\left(1+\lambda\right)d\lambda,\\
    \label{Eq:ACCDilationSpectrumC}
	    &=&\mathord{-}\mu_{\gamma}(H+\imaginary\,\omega;\gammabar)\,		    
	        \Gamma(H+\imaginary\,\omega)\,
	            \Gamma(\mathord{-}H\mathord{-}\imaginary\,\omega),{~~~~~~}
\end{eqnarray}
\end{subequations}
for Hurst's exponent $0\!<\!H\!<\!1$, where $\mu_{\gamma}(n;\gammabar)=\mathbb{E}[\gamma^{n}]$ denotes the $n$th moment of the instantaneous \ac{SNR}. Similarly, we obtain the \ac{LDS} of the \ac{ACR} using \cite[Eq.~(2.10.2/1)]{BibPrudnikovBookVol2} as
\begin{subequations}
    \label{Eq:ACRDilationSpectrum}
    \setlength\arraycolsep{1.4pt}
    \begin{eqnarray}
        \label{Eq:ACRDilationSpectrumA}
        \mathcaltilde{Q}_{avg}\left(\omega,\gammabar\right)
            &=&\mathfrak{F}_{\lambda}\!\bigl\{
                    \mathfrak{L}^{-1}_{H,\gammabar}\bigl\{
                        \mathcal{Q}_{avg}\left(\gammabar\right)
                    \bigr\}(\lambda)
                \bigl\}\left({\omega}\right),\\
        \label{Eq:ACRDilationSpectrumB}
            &=&\mu_{\gamma}(H+\imaginary\,\omega;\gammabar)\!
		        \int_{0}^{\infty}\!\frac{\lambda^{-\imaginary\,\omega}}{\lambda^{H+1}}
		            \frac{\widehat\Gamma(b,a\lambda)}{\Gamma(b)}
		            d\lambda,\\                
        \label{Eq:ACRDilationSpectrumC}
            &=&\mu_{\gamma}(H+\imaginary\,\omega;\gammabar)
                \frac{a^{H+\imaginary\,\omega}\Gamma(b-H-\imaginary\,\omega)}
                    {\Gamma(b)\,(H+\imaginary\,\omega)},{~~~~~}
    \end{eqnarray}
\end{subequations}
for Hurst's exponent $0\!<\!H\!<\!b$. In order to find $\mathcal{Z}_{a,b}(u)$, we carry the ratio of $\mathcaltilde{C}_{avg}\left(\omega,\gammabar\right)$ to $\mathcaltilde{Q}_{avg}\left(\omega,\gammabar\right)$, that is  
\begin{equation}
\label{Eq:DilationSpectrumRatioForWojnarABEPAndACC}
\frac{\mathcaltilde{C}_{avg}\left(\omega,\gammabar\right)}{\mathcaltilde{Q}_{avg}\left(\omega,\gammabar\right)}
    =\Gamma(b)\frac{\Gamma(H+\imaginary\,\omega)\,
        \Gamma(1-H-\imaginary\,\omega)}
			{a^{H+\imaginary\,\omega}\,\Gamma(b-H-\imaginary\,\omega)},
\end{equation}
which certainly exists for $0\!<\!H\!<\!\min(1,b)$. Substituting \eqref{Eq:DilationSpectrumRatioForWojnarABEPAndACC} into \eqref{Eq:AuxiliaryFunctionForRelationshipBetweenAPMs}
yields
\ifCLASSOPTIONtwocolumn
\begin{multline}
\label{Eq:AuxiliaryFunctionForWojnarABEPAndACC}
\mathcal{Z}_{a,b}(u)=
	\frac{\Gamma(b)}{u}
	\frac{1}{2\pi}\!
		\int_{-\infty}^{+\infty}
			\Gamma(H-\imaginary\,\omega)\manualtimesbreak
				\frac{\Gamma(1-H+\imaginary\,\omega)}{\Gamma(b-H+\imaginary\,\omega)}
					{(au)}^{-H+\imaginary\,\omega}\,			
						d\omega,
\end{multline}
\else
\begin{equation}
\label{Eq:AuxiliaryFunctionForWojnarABEPAndACC}
\mathcal{Z}_{a,b}(u)=
	\frac{\Gamma(b)}{u}
	\frac{1}{2\pi}\!
		\int_{-\infty}^{+\infty}
			\Gamma(H-\imaginary\,\omega)
				\frac{\Gamma(1-H+\imaginary\,\omega)}{\Gamma(b-H+\imaginary\,\omega)}
					{(au)}^{-H+\imaginary\,\omega}\,			
						d\omega,
\end{equation}
\fi
which absolutely converges for $0\!<\!H\!<\!\min(1,b)$. Finally, changing the variable
$H\mathord{-}\imaginary\,\omega\rightarrow{s}$ simplifies \eqref{Eq:AuxiliaryFunctionForWojnarABEPAndACC} 
into the contour integral of Kummer's confluent hypergeometric
function \cite[Eq.~(07.20.07.0004.01)]{BibWolfram2010Book}, which proves \theoremref{Theorem:RelationshipBetweenACCAndABEPForBinaryModulationSchemes}.
\end{proof}
It is accordingly worth mentioning that,~from~a~broader~perspective, \theoremref{Theorem:RelationshipBetweenACCAndABEPForBinaryModulationSchemes} reveals an intimate connection between information theory and estimation theory and may establish and superimpose a large set of new and innovative ideas,~a~few of which are presented below, in the theory of wireless communications.

\subsection{Special Cases of \theoremref{Theorem:RelationshipBetweenACCAndABEPForBinaryModulationSchemes}}
\label{Section:RelationshipBetweenACCandABEP:SpecialCases}
It is worth reviewing the special cases of \theoremref{Theorem:RelationshipBetweenACCAndABEPForBinaryModulationSchemes} for convenience and clarity. For coherent signalling using binary modulation schemes\cite{BibWojnarTCOM1986}, we set $b=1/2$ in \eqref{Eq:AuxiliaryFunctionForACCAndWojnarABEP}, and then obtain the auxiliary function as  $\mathcal{Z}_{a,1/2}(u)\!=\!\frac{1}{u}\bigl(1\mathord{-}2\sqrt{au}\,\DawsonF{\sqrt{au}}\bigr)$, where $\DawsonF{x}\!=\!\exp(-x^2)\int_{0}^{x}\exp(u^2)du$ denotes the Dawson's integral \cite[Eq. (7.1.16)]{BibAbramowitzStegunBook}. Accordingly, substituting this result into \eqref{Eq:RelationshipBetweenACCAndABEPForBinaryModulationSchemes}, we obtain the \ac{ACC} $\mathcal{C}_{avg}(\gammabar)$ in terms of the \ac{ABER}  $\mathcal{E}_{avg}(\gammabar)$ of binary modulation schemes, that is 
\ifCLASSOPTIONtwocolumn
\begin{multline}
\label{Eq:ACCObtainedFromABEPForCoherentBinaryModulations}
\mathcal{C}_{avg}\left(\gammabar\right)=
	\int_{0}^{\infty}\frac{1}{u}\Bigl(1\mathord{-}2\sqrt{au}\,\DawsonF{\sqrt{au}}\Bigr)
	    \manualtimesbreak
		\Bigl\{1\mathord{-}2\,\mathcal{E}_{avg}\left(u\,\gammabar\right)\Bigr\}du,
\end{multline}
\else
\begin{equation}
\label{Eq:ACCObtainedFromABEPForCoherentBinaryModulations}
\mathcal{C}_{avg}\left(\gammabar\right)=
	\int_{0}^{\infty}\frac{1}{u}\Bigl(1\mathord{-}2\sqrt{au}\,\DawsonF{\sqrt{au}}\Bigr)
		\Bigl\{1\mathord{-}2\,\mathcal{E}_{avg}\left(u\,\gammabar\right)\Bigr\}du,
\end{equation}
\fi
where $a=1/2$ and $a=1$ are for the coherent orthogonal binary FSK (BFSK) and the coherent antipodal binary PSK (BPSK), respectively. In addition, the other special case is for non-coherent binary modulation schemes. Herewith, substituting $b=1$ in \eqref{Eq:RelationshipBetweenACCAndABEPForBinaryModulationSchemes} yields $\mathcal{Z}_{a,1}\left(u\right)={\exp(-au)}/{u}$, and then we obtain $\mathcal{C}_{avg}(\gammabar)$ in terms of the \ac{ABER} $\mathcal{E}_{avg}(\gammabar)$ of non-coherent signalling using binary modulation schemes, that is 
\begin{equation}\label{Eq:ACCObtainedFromABEPForNoncoherentBinaryModulations}
\mathcal{C}_{avg}\left(\gammabar\right)=
	\int_{0}^{\infty}\frac{{e}^{-au}}{u}
			\Bigl\{1-2\,\mathcal{E}_{avg}\left(u\,\gammabar\right)\Bigr\}du,
\end{equation}
where $a=1/2$ and $a=1$ are given for the orthogonal non-coherent FSK (NCFSK) and, the antipodal differentially coherent PSK (BDPSK) binary modulation schemes, respectively.

\subsection{\ac{ACC} Analysis Using Closed-Form \ac{ABER} Expressions}
\label{Section:RelationshipBetweenACCandABEP:AccAnalysisUsingABER}
Without having to acknowledge the \ac{SNR} distribution and the underlying \ac{SNR} details and various \ac{SNR} settings of wireless communications, we can determine~its~\ac{ACC} performance from exploiting its \ac{ABER} performance for binary modulation schemes. For analytical correctness and completeness, take for example a communications systems using binary modulated signalling over Rayleigh fading channels for which the closed-form \ac{ABER} $\mathcal{E}_{avg}\left(\gammabar\right)$ is given by $\mathcal{E}_{avg}\left(\gammabar\right)\!=\!\frac{1}{2}-\frac{1}{2}\bigl(\frac{a\gammabar}{1+a\gammabar}\bigr)^{b}$\cite[Eq.~(16)]{BibAaloPiboongungonIskander2005}. Using \cite[Eq.~(8.4.2/5)]{BibPrudnikovBookVol3}, we express it in terms of Meijer's G function as follows 
\begin{equation}
    \label{Eq:WojnarABEPInRayleighEnvironments}
    \mathcal{E}_{avg}\left(\gammabar\right)=
            \frac{1}{2}-\frac{1}{2\Gamma(b)}\MeijerG[right]{1,1}{1,1}{a\gammabar}{1}{b}.
\end{equation}
With the aid of \theoremref{Theorem:RelationshipBetweenACCAndABEPForBinaryModulationSchemes}, wherein~we~replace~both~\eqref{Eq:WojnarABEPInRayleighEnvironments}~and $\Hypergeom{1\!}{1}{1}{b}{-au}\!=\!\MeijerG[right]{1,1}{1,2}{au}{0}{0,1-b}$\cite[Eq.\!~(07.20.26.0006.01)]{BibWolfram2010Book}, where $\MeijerGDefinition{m,n}{p,q}{\cdot}$ is the Meijer's G function \cite[Eq.~(8.3.22)]{BibPrudnikovBookVol3}, we express the \ac{ACC} for Rayleigh fading~channels~as 
\begin{equation}\label{Eq:ACCUsingWojnarABEPInRayleighEnvironmentsII}
\mathcal{C}_{avg}\left(\gammabar\right)=
		\int_{0}^{\infty}\!\frac{1}{u}
	        \MeijerG[right]{1,1}{1,2}{{a}{u}}{0}{0,1\mathord{-}b}
			\MeijerG[right]{1,1}{1,1}{au\,\gammabar}{1}{b}du,
\end{equation}
where performing some simple algebraic manipulations\cite[Eqs. (2.24.2/2) and (8.4.16/14)]{BibPrudnikovBookVol3} and then utilizing\cite[Eq. (6.5.15)]{BibAbramowitzStegunBook} yields the \ac{ACC} for Rayleigh fading channels, that is 
\begin{equation}
    \label{Eq:ChannelCapacityOverRayleighChannels}
    C_{avg}\left(\gammabar\right)=
	    \exp\left(\frac{1}{\gammabar}\right)
	        \ExpIntegralE{1}{\frac{1}{\gammabar}},
\end{equation}
which is in perfect agreement with \cite[Eq.~(34)]{BibAlouiniGoldsmithTVT1999} as expected, and where $\ExpIntegralE{n}{\cdot}$ denotes the exponential integral \cite[Eq.(5.1.4)]{BibAbramowitzStegunBook}. Even if $\mathcal{C}_{avg}\left(\gammabar\right)$ is specifically obtained with the aid of $\mathcal{E}_{avg}\left(\gammabar\right)$, it does not depend on the modulation parameters $a$ and $b$ as expected and demonstrated in \eqref{Eq:ChannelCapacityOverRayleighChannels}. The number of such examples for analytical correctness and completeness can be extended easily considering other pairs of \ac{ACC} and \ac{ABER} available in the literature. 

\ifCLASSOPTIONtwocolumn
\begin{figure*}[tp] 
\centering
\begin{subfigure}{0.98\columnwidth}
    \centering
    \includegraphics[clip=true, trim=0mm 0mm 0mm 0mm,width=1.0\columnwidth,height=0.80\columnwidth]{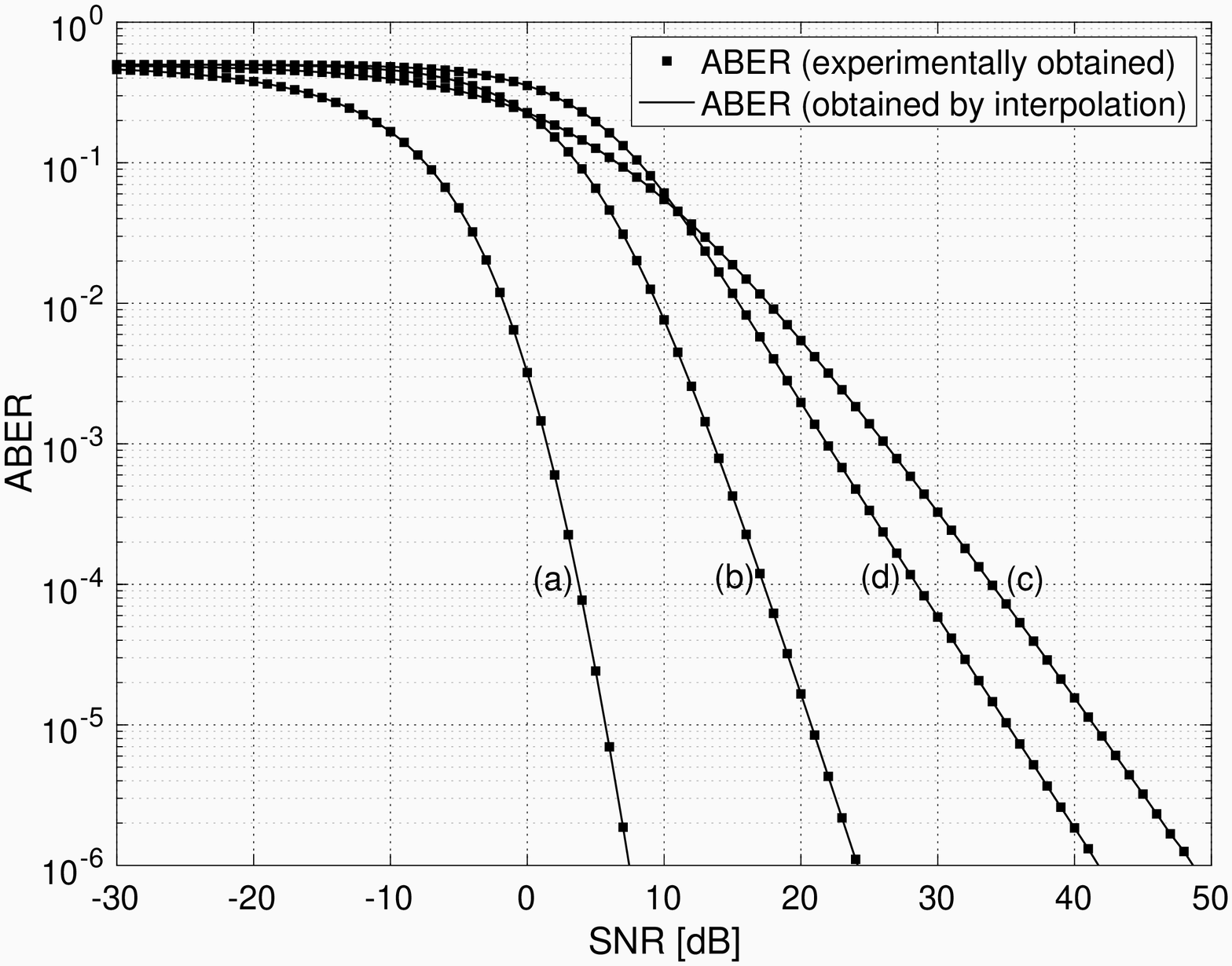}
    \caption{The \ac{ABER} curves of four different communications systems signaling using BPSK modulation in generalized fading environments. Referring to \secref{Section:RelationshipBetweenACCandABEP:AccAnalysisUsingEmpiricalABER}, we choose $N=201$ and $\gammabar_{B}=100$ dB.}
    \label{Fig:ABEPForMRCReceiver}
\end{subfigure}
~~~
\begin{subfigure}{0.98\columnwidth}
    \centering
    \includegraphics[clip=true, trim=0mm 0mm 0mm 0mm,width=1.0\columnwidth,height=0.80\columnwidth]{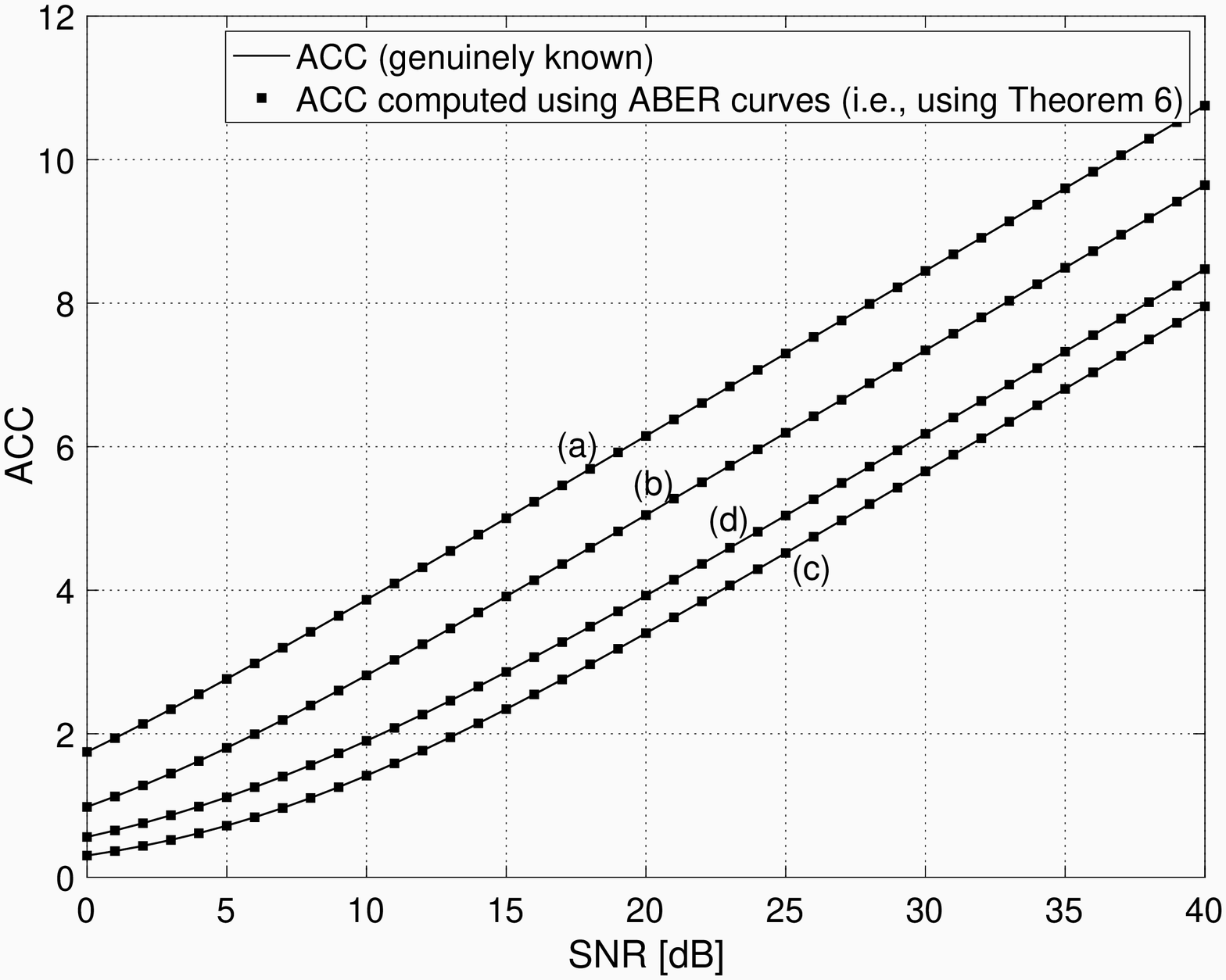}
    \caption{The genie-aided \ac{ACC} curves and the \ac{ACC} results obtained from the \ac{ABER} curves in generalized fading environments.\vspace{2mm}}
    \label{Fig:ACCForMRCReceiver}
\end{subfigure}
\caption{The numerical illustration of how the \ac{ACC} of any communications system is readily obtained from its \ac{ABER} curves for various \ac{SNR} settings: $(a),\,(b),\,(c),\,(d)$ that are obviously not accounted in the numerical computation.}
\label{Fig:RelationshipBetweenACCAndABEPForMRCReceiver}
\end{figure*}
\else
\begin{figure*}[tp] 
\centering
\begin{subfigure}{0.45\columnwidth}
    \centering
    \includegraphics[clip=true, trim=0mm 0mm 0mm 0mm,width=1.0\columnwidth,height=0.80\columnwidth]{./images/image_01_a.eps}
    \caption{The \ac{ABER} curves of four different communications systems signaling using BPSK modulation in generalized fading environments. Referring to \secref{Section:RelationshipBetweenACCandABEP:AccAnalysisUsingEmpiricalABER}, we choose $N=201$ and $\gammabar_{B}=100$ dB.}
    \label{Fig:ABEPForMRCReceiver}
\end{subfigure}
~~
\begin{subfigure}{0.45\columnwidth}
    \centering
    \includegraphics[clip=true, trim=0mm 0mm 0mm 0mm,width=1.0\columnwidth,height=0.80\columnwidth]{./images/image_01_b.eps}
    \caption{The genie-aided \ac{ACC} curves and the \ac{ACC} results obtained from the \ac{ABER} curves in generalized fading environments.}
    \label{Fig:ACCForMRCReceiver}
\end{subfigure}
\caption{The numerical illustration of how the \ac{ACC} of any communications system is readily obtained from its \ac{ABER} curves for various \ac{SNR} settings: $(a),\,(b),\,(c),\,(d)$ that are obviously not accounted in the numerical computation.}
\label{Fig:RelationshipBetweenACCAndABEPForMRCReceiver}
\vspace{-3mm} 
\end{figure*}
\fi

\subsection{\ac{ACC} Analysis Using Empirical \ac{ABER} Measurements}
\label{Section:RelationshipBetweenACCandABEP:AccAnalysisUsingEmpiricalABER}
At the design and implementation stages~of~communications systems, it is crucial to empirically evaluate \ac{ACC} to investigate the desired reliability of information transmission and to obtain conceptual results from practical and empirical perspectives. On the basis thereof and of using \theoremref{Theorem:RelationshipBetweenACCAndABEPForBinaryModulationSchemes}, we can readily calculate the \ac{ACC} from the empirical \ac{ABER} measurements for a binary modulation scheme. Let us assume that we obtain the measurement set $\mathcal{S}_N\!=\!\{(\gammabar_{1},\mathcal{E}_{1}),(\gammabar_{2},\mathcal{E}_{2}),\ldots,(\gammabar_{N},\mathcal{E}_{N})\}$, where we have $\mathcal{E}_{n}\!=\!\mathcal{E}_{avg}\left(\gammabar_{n}\right)$ for all $n\in\{1,2,\ldots,N\}$. Accordingly, using an efficient interpolation technique\cite{BibDavisBook1975,BibPressTeukolskyNumericalRecipes1992Book,BibWong2010Book}, we can accurately approximate $\mathcal{E}_{avg}\left(\gammabar\right)$ as $\mathcal{E}_{avg}\left(\gammabar\right)\approx\mathcal{E}_{int}\left(\gammabar\,|\,\mathcal{S}_N\right)$ for any average \ac{SNR} $\gammabar\in\mathbb{R}_{+}$. As well explained in \secref{Section:RelationshipsAmongAPMs:EmpiricalAPMEstimation}, we can efficiently apply the Lagrange's interpolation to $\mathcal{S}_N$, and then we can write $\mathcal{E}_{int}\left(\gammabar\,|\,\mathcal{S}_N\right)$ as
\begin{equation}\label{Eq:ABEPApproximatedByInterpolation}
    \mathcal{E}_{int}\left(\gammabar\,|\,\mathcal{S}_N\right)
        =\sum_{n=1}^{N}
            \mathcal{E}_{n}
            \prod_{\substack{k=1,\,{k}\neq{n}}}^{N}
                \frac{\gammabar-\gammabar_k}{\gammabar_{n}-\gammabar_k}.
\end{equation}
Consequently, substituting \eqref{Eq:ABEPApproximatedByInterpolation} into \eqref{Eq:RelationshipBetweenACCAndABEPForBinaryModulationSchemes}, we can estimate the \ac{ACC} $\mathcal{C}_{avg}\left(\gammabar\right)$ from the empirical \ac{ABER} measurements. For example, some \ac{ABER} curves of communications systems are depicted in \figref{Fig:ABEPForMRCReceiver} for BPSK signalling in generalized fading environments, and the corresponding genie-aided \ac{ACC} curves are depicted in \figref{Fig:ACCForMRCReceiver}.
Fortunately, with the aid of \theoremref{Theorem:RelationshipBetweenACCAndABEPForBinaryModulationSchemes}, and without the need for the additional information for those \ac{ABER} curves, we calculate each \ac{ACC} curve from the corresponding \ac{ABER} curve and plot in \figref{Fig:ACCForMRCReceiver} the calculated \ac{ACC} curves. Further, in \figref{Fig:ACCForMRCReceiver}, we notice that the calculated \ac{ACC} curves are in perfect agreement with the genie-aided curves as expected. In addition, we illustrate in \figref{Fig:ACCVersusABEPForMRCReceiver} the \ac{ACC} versus the \ac{ABER}. We notice therein the matter fact that \ac{ACC} is a linear function of the logarithm of \ac{ABER} when the \ac{ABER} approaches zero (i.e., when the average \ac{SNR} $\gammabar$ increases). However, this linearity is impaired for low average\ac{SNR} values. From this point of view, using the definition of low- and high-\ac{SNR} regimes\cite{BibYilmazTUBITAK2019}, we can conclude that the \ac{ACC} of any communications system is determined by its \ac{ABER} in low-\ac{SNR} regime rather than in high-\ac{SNR} regime. 

\ifCLASSOPTIONtwocolumn
\begin{figure}[tp] 
\centering
\includegraphics[clip=true, trim=0mm 0mm 0mm 0mm,width=1.0\columnwidth,height=0.80\columnwidth]{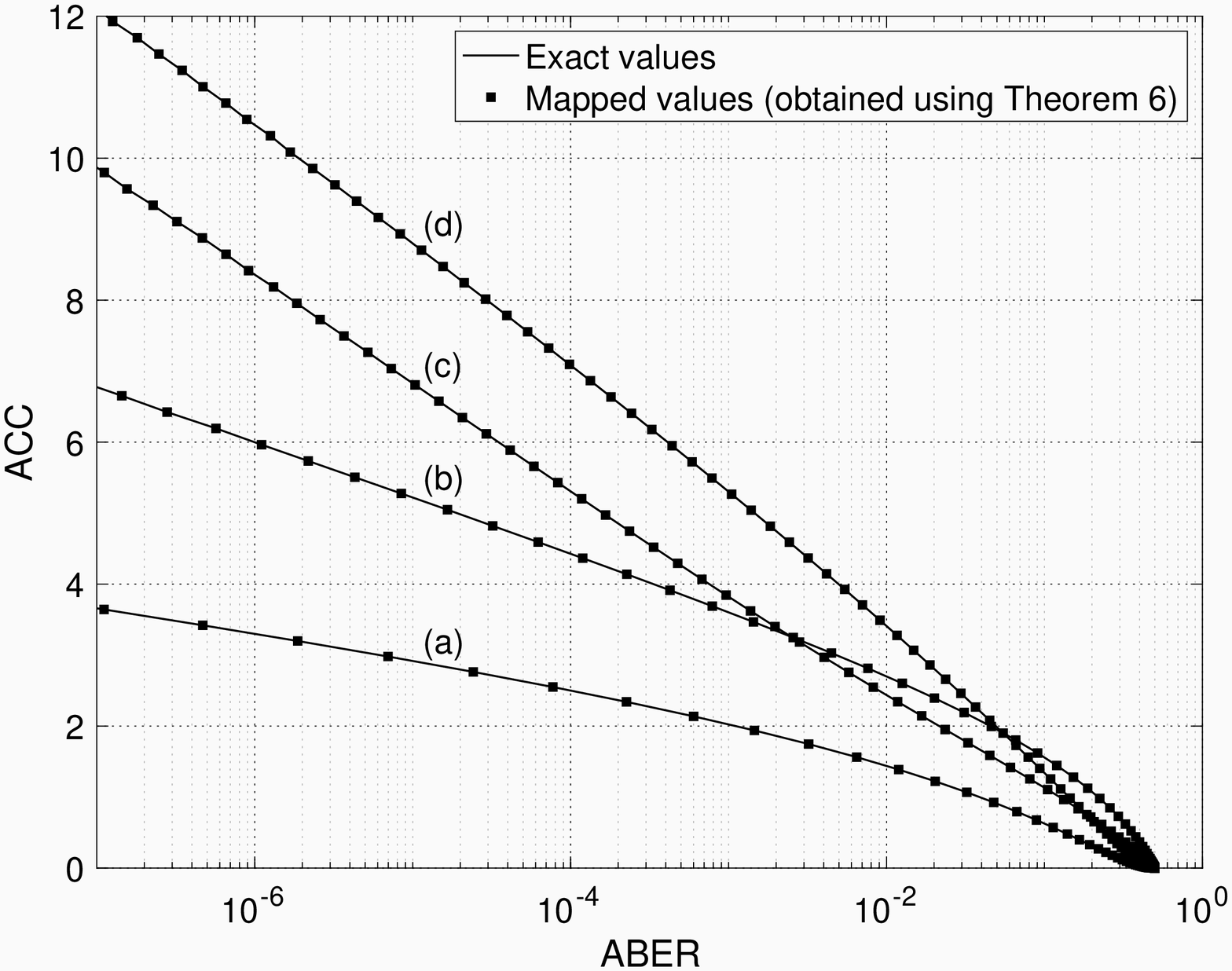}
\caption{The relationship between the \ac{ACC} and the \ac{ABER} for a wireless communication system whose \ac{ABER} curves are given in \figref{Fig:ABEPForMRCReceiver} for various \ac{SNR} settings.}
\label{Fig:ACCVersusABEPForMRCReceiver}
\end{figure}
\else
\begin{figure}[tp] 
\centering
\includegraphics[clip=true, trim=0mm 0mm 0mm 0mm,width=0.45\columnwidth,height=0.38\columnwidth]{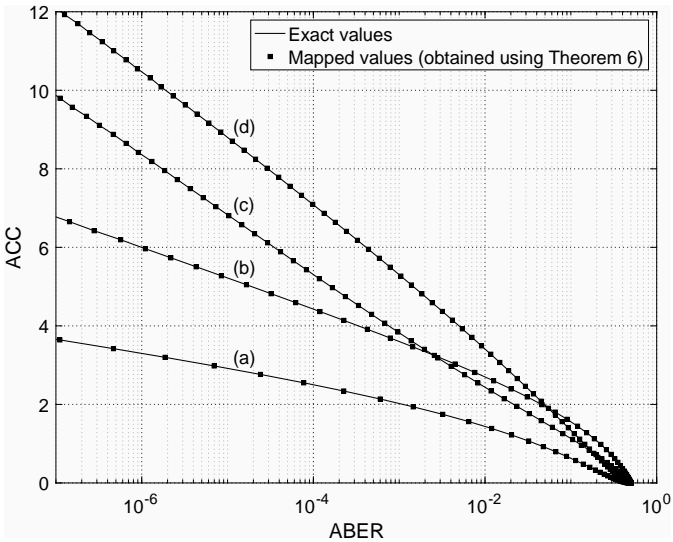}
\caption{The relationship between the \ac{ACC} and the \ac{ABER} for a wireless communication system whose \ac{ABER} curves are given in \figref{Fig:ABEPForMRCReceiver} for various \ac{SNR} settings.}
\label{Fig:ACCVersusABEPForMRCReceiver}
\vspace{-3mm} 
\end{figure}
\fi


\section{\ac{OP} and \ac{OC} Analyses Using \ac{ABER}}
\label{Section:RelationshipBetweenAOPandAOCandACC}
Particularly when the fading conditions vary slowly, a possible deep fading has the potential to affect many successive bits during information transmission, resulting in large error bursts. If these errors cannot be corrected even if very complex coding and detection schemes are employed, the transmission channel will become unusable. In that context, the \ac{OP}, denoted by $\mathcal{P}_{out}(\gammabar;\gamma_{th})$, is a channel quality measure\cite{BibAlouiniBook,BibGoldsmithBook}, defined as the probability that the \ac{SNR} $\gamma$ drops below a certain threshold $\gamma_{th}$ for a certain average \ac{SNR} $\gammabar$, that is defined as 
\begin{equation}\label{Eq:AOPDefinition}
    \mathcal{P}_{out}(\gammabar;\gamma_{th})=
        \Expected{\theta\!\left(\gamma_{th}-\gamma\right)},
\end{equation}
for $\gamma\in\mathbb{R}^{+}$ and $\gamma_{th}\in\mathbb{R}^{+}$, where $\theta\!\left(\cdot\right)$ is Heaviside's theta function (i.e., unit-step function) defined in \cite[Eq.~(1.8.3)]{BibZwillingerBook},\cite[Eq.~(14.05.07.0002.01)]{BibWolfram2010Book}, that is,
\begin{equation}\label{Eq:HeavisideThetaFunction}
    \theta\!\left(\gamma_{th}-\gamma\right)=
        \textstyle
        \begin{cases}
            1           & \gamma<\gamma_{th} \\
            \frac{1}{2} & \gamma=\gamma_{th} \\
            0           & \gamma>\gamma_{th}
        \end{cases}.
\end{equation}
Note that the exact \ac{OP} analysis in \eqref{Eq:AOPDefinition}, i.e., $\mathcal{P}_{out}(\gammabar;\gamma_{th})$ requires the \ac{SNR} distribution. However, we demonstrate how we obtain the \ac{OP} of any communications system by using its exact \ac{ACC} performance, in particular without the knowledge of the \ac{SNR} distribution and the underlying \ac{SNR} settings.



\begin{theorem}[\ac{OP} analysis using \ac{ACC}]\label{Theorem:RelationshipBetweenAOPAndACC}
The \ac{OP} $\mathcal{P}_{out}(\gammabar;\gamma_{th})$ of a wireless communications system is obtained using its exact (not approximate) \ac{ACC} $\mathcal{C}_{avg}(\gammabar)$, that is
\begin{equation}\label{Eq:RelationshipBetweenAOPAndACC}
    \mathcal{P}_{out}(\gammabar;\gamma_{th})=
        1-\frac{1}{\pi}\Im\biggl\{
            \mathcal{C}_{avg}\left(\!-\frac{\gammabar}{\gamma_{th}}\right)\biggr\},
\end{equation}
for a certain threshold $\gamma_{th}\in\mathbb{R}^{+}$, where $\ImagPart{\cdot}$ gives~the~imaginary part of its argument. Further, for numerical stability,  $\mathcal{C}_{avg}(-\gammabar)\approx\mathcal{C}_{avg}(e^{\varepsilon+\imaginary\pi}{\gammabar})$, where we choose $\varepsilon\approx{0.0001}$.
\ifCLASSOPTIONtwocolumn
\fi
\end{theorem}

\begin{proof}
Due to the discordance between the monotonic natures of \ac{OP} and \ac{ACC} (i.e. since $\mathcal{P}_{out}(\gammabar;\gamma_{th})\geq\mathcal{P}_{out}(\gammabar+\Delta\gammabar;\gamma_{th})$ and $\mathcal{C}_{avg}(\gammabar)\leq\mathcal{C}_{avg}(\gammabar+\Delta\gammabar)$ for all $\Delta\gammabar\in\mathbb{R}_{+}$), both the \ac{OP} and the \ac{ACC} together do not preserve the same scaling order according to \theoremref{Theorem:RelationshipExistenceBetweenAPMs}, and thus their \ac{LDS}~spectrums do not exist for a mutually common Hurst's exponent. But, we write 
\begin{equation}\label{Eq:AOPUsingComplementaryAOP}
    \mathcal{P}_{out}(\gammabar;\gamma_{th})=
        1-\widehat{\mathcal{P}}_{out}(\gammabar;\gamma_{th}),
\end{equation}
where $\widehat{\mathcal{P}}_{out}(\gammabar;\gamma_{th})$ is the complementary \ac{OP}, defined as 
\begin{equation}\label{Eq:ComplementaryAOPDefinition}
    \widehat{\mathcal{P}}_{out}(\gammabar;\gamma_{th})=
        \Expected{\theta\!\left(\gamma-\gamma_{th}\right)},   
\end{equation}
which is monotonically increasing with respect to the average \ac{SNR} $\gammabar$ (i.e., $\widehat{\mathcal{P}}_{out}(\gammabar;\gamma_{th})\leq\widehat{\mathcal{P}}_{out}(\gammabar+\Delta\gammabar;\gamma_{th})$). Therefore, the \ac{LDS} of $\widehat{\mathcal{P}}_{out}(\gammabar;\gamma_{th})$ and that of $\mathcal{C}_{avg}(\gammabar)$ exist~for a mutually common Hurst's exponent. As is required for \theoremref{Theorem:RelationshipExistenceBetweenAPMs}, we have already obtained in \eqref{Eq:ACCDilationSpectrumC} the \ac{LDS} of $\mathcal{C}_{avg}(\gammabar)$ for Hurst's exponent $0\!<\!H\!<\!1$, and moreover using \cite[Eq. (2.2.1/1)]{BibPrudnikovBookVol4}, we obtain the \ac{LDS} of the complementary \ac{OP} as  
\begin{subequations}
    \label{Eq:ComplementaryAOPDilationSpectrum}
    \setlength\arraycolsep{1.4pt}
    \begin{eqnarray}
    \label{Eq:ComplementaryAOPDilationSpectrumA}
    \mathcaltilde{P}_{out}\left(\omega,\gammabar\right)&=&
            \mathfrak{F}_{\lambda}\!\bigl\{\mathfrak{L}^{-1}_{H,\gammabar}
                \bigl\{\widehat{\mathcal{P}}_{out}\left(\gammabar\right)\bigr\}(\lambda)\bigl\}\left({\omega}\right),\\
    \label{Eq:ComplementaryAOPDilationSpectrumB}
        &=&\mu_{\gamma}(H+\imaginary\,\omega;\gammabar)\!
		    \int_{0}^{\infty}\!\frac{\lambda^{-\imaginary\,\omega}}{\lambda^{H+1}}\,
		        \theta\left(\lambda-\gamma_{th}\right)d\lambda,{~~~~~~}\\
    \label{Eq:ComplementaryAOPDilationSpectrumC}
	    &=&\mu_{\gamma}(H+\imaginary\,\omega;\gammabar)\,\gamma_{th}^{-H-\imaginary\omega}
	        {\left(H+\imaginary\,\omega\right)}^{-1},
\end{eqnarray}
\end{subequations}
for Hurst's exponent $0\!<\!H\!<\!\infty$, where $\mu_{\gamma}(n;\gammabar)\!=\!\mathbb{E}[\gamma^{n}]$ is the $n$th moment of the \ac{SNR} distribution. Referring to \theoremref{Theorem:RelationshipExistenceBetweenAPMs}, we find $\mathcal{Z}(u)$ by the ratio of \eqref{Eq:ComplementaryAOPDilationSpectrumC} to \eqref{Eq:ACCDilationSpectrumC}, and then using \cite[Eq.~(06.05.16.0010.01)]{BibWolfram2010Book}, we simplify it more to 
\begin{equation}
    \frac{\mathcaltilde{P}_{out}\left(\omega,\gammabar\right)}	          
        {\mathcaltilde{C}_{avg}\left(\omega,\gammabar\right)}=
            -\frac{1}{\pi}\sin(\pi(H+i\,\omega)),
\end{equation}
which exists for Hurst's exponent $0<H<1$. Replacing this ratio in \eqref{Eq:AuxiliaryFunctionForRelationshipBetweenAPMs} and changing the variable $H-\imaginary\omega\rightarrow{s}$, we obtain 
\begin{equation}
    \label{Eq:AuxiliaryFunctionForACCAndComplementaryAOP}
    \mathcal{Z}(u)=
        \frac{1}{\pi}\Im\biggl\{
            \frac{1}{2\pi\imaginary}
                \int_{-H-\imaginary\,\infty}^{-H+\imaginary\,\infty}
                    u^{s-1}
                    \bigl(-\gamma_{th}\bigr)^{s}ds
        \biggr\},
\end{equation}
for $0\,\mathord{<}\,H\,\mathord{<}\,1$. Substituting \eqref{Eq:AuxiliaryFunctionForACCAndComplementaryAOP} into \eqref{Eq:RelationshipBetweenAPMs}, and therein changing the order of integrals and then exploiting \ac{MT}, we obtain
\begin{equation}\label{Eq:ComplementaryAOPDefinitionII}
    \widehat{\mathcal{P}}_{out}(\gammabar;\gamma_{th})=
        \frac{1}{\pi}\Im\Bigl\{
            \mathcal{C}_{avg}\left(\!-{\gammabar}/{\gamma_{th}}\right)\Bigr\}.
\end{equation}
Finally, exercising \eqref{Eq:AOPUsingComplementaryAOP} and \eqref{Eq:ComplementaryAOPDefinitionII} together, we find \eqref{Eq:RelationshipBetweenAOPAndACC}, which completes the proof of \theoremref{Theorem:RelationshipBetweenAOPAndACC}.
\end{proof}

It is worth recalling that the \ac{OC}, denoted by $\mathcal{C}_{out}(\gammabar;\gamma_{th})$, is defined as the probability that the instantaneous~\ac{CC}~falls below a certain information rate $C_{th}$. As such, the \ac{CC} while targeting a certain \ac{OP} is called the \ac{OC} performance, that is
\begin{subequations}
    \label{Eq:AOCDefinition}
    \setlength\arraycolsep{1.4pt}
    \begin{eqnarray}
        \label{Eq:AOCDefinitionA}
        \mathcal{C}_{out}(\gammabar;C_{th})
            &\trigeq&\Expected{\theta\!\left(C_{th}-\log(1+\gamma)\right)},\\
        \label{Eq:AOCDefinitionB}
	        &=&\mathcal{P}_{out}(\gammabar;e^{C_{th}}-1).
\end{eqnarray}
\end{subequations}

\begin{theorem}[\ac{OC} analysis using \ac{ACC}]
\label{Theorem:RelationshipBetweenAOCAndACC}
The \ac{OC} $\mathcal{C}_{out}(\gammabar;\gamma_{th})$ of a wireless communications system is obtained using its exact (not approximate) \ac{ACC} $\mathcal{C}_{avg}(\gammabar)$, that is
\begin{equation}\label{Eq:RelationshipBetweenAOCAndACC}
    \mathcal{C}_{out}(\gammabar;C_{th})=
        1-\frac{1}{\pi}\Im\Bigl\{
                \mathcal{C}_{avg}\left(-{\gammabar}/{(e^{C_{th}}-1)}\right)
            \Bigr\}
\end{equation}
for a certain capacity threshold $C_{th}\,\mathord{\in}\,\mathbb{R}^{+}$.
\end{theorem}

\begin{proof}
The proof is obvious using \eqref{Eq:AOCDefinitionB} in \theoremref{Theorem:RelationshipBetweenAOPAndACC}.
\end{proof}

\subsection{\ac{OP} and \ac{OC} Analyses Using Exact \ac{ACC} Expressions}
\label{Section:RelationshipBetweenAOPandAOCandACC:OPOCAnalysesUsingACCExpressions}
The two novel relationships, which are proposed in \theoremref{Theorem:RelationshipBetweenAOPAndACC} and \theoremref{Theorem:RelationshipBetweenAOCAndACC}, respectively, provide new insights into the \ac{OP} and \ac{OC} analyses of any communications system and show that an exact (non-approximate) ACC is sufficient for \ac{OP} and \ac{OC} analyses. Since neither integration~nor~further~mathematical operations are required, these novel relationships, i.e., \eqref{Eq:RelationshipBetweenAOPAndACC} and \eqref{Eq:RelationshipBetweenAOCAndACC}, are surely considered as closed-form expressions. However, in order to demonstrate their correctness, that is to derive~the~\ac{OP}~expressions from the closed-form \ac{ACC} expressions in the literature, the hindrance in algebraic (mathematical) simplification is the elimination of imaginary-part operation. In more details, in the literature, the real-part and imaginary-part operations are not known to be expressed in terms of integral transforms. If they were known, the real-part and imaginary-part of higher transcendental functions (e.g., hypergeometric, Meijer's G and Fox's H functions) would be simplified to the expressions having no the real-part and imaginary-part operations.~To~the~best of our knowledge, expressing real-part and imaginary-part operations in terms of integral transforms has so far not been reported in the literature. However, thanks to \theoremref{Theorem:RelationshipBetweenAPMs}, we have rewritten them in terms of integral transforms in the following theorem.

\begin{theorem}[Real-part and imaginary-part operations]
\label{Theorem:RealAndImaginaryParts} 
For a real-valued \ac{APM}\footnote{A real-valued \ac{APM} is the one that assigns average \acp{SNR} to real numbers.} $\mathcal{G}_{avg}(\gammabar)$, $\Im\bigl\{\mathcal{G}_{avg}(-\gammabar)\bigr\}$ and $\Re\bigl\{\mathcal{G}_{avg}(-\gammabar)\bigr\}$ are respectively given by 
\setlength\arraycolsep{1.4pt}
\begin{eqnarray}
    \label{Eq:RelationshipBetweenImaginaryPartAPMAndAPM} 
    \Im\bigl\{\mathcal{G}_{avg}(-\gammabar)\bigr\}&=&
    \pi\,
    \InvMellinTransform{s}{
        \frac{\mathfrak{M}_{\gammabar}\!\bigl\{\mathcal{G}_{avg}(\gammabar)\bigr\}\!\left(s\right)}
            {\Gamma\left(1+s\right)\Gamma\left(-s\right)}
    }{\gammabar},\\
    \label{Eq:RelationshipBetweenRealPartAPMAndAPM}     
    \Re\bigl\{\mathcal{G}_{avg}(-\gammabar)\bigr\}&=&
    \pi\,
    \InvMellinTransform{s}{
        \frac{\mathfrak{M}_{\gammabar}\!\bigl\{\mathcal{G}_{avg}(\gammabar)\bigr\}\!\left(s\right)}
            {\Gamma\left(\frac{1}{2}+s\right)\Gamma\left(\frac{1}{2}-s\right)}
    }{\gammabar},{~~~}
\end{eqnarray}
where $\MellinTransform{\cdot}{\cdot}{\cdot}$ is the \ac{MT}, and $\InvMellinTransform{\cdot}{\cdot}{\cdot}$ is the \ac{IMT}\footnoteref{Footnote:MellinTransform}.
\ifCLASSOPTIONtwocolumn
\pagebreak[4]
\fi
\end{theorem}

\begin{proof}
The proof is omitted due to space limitations.
\end{proof}

Although \theoremref{Theorem:RelationshipBetweenAOCAndACC} is itself in closed form, let us consider some \ac{OP} and \ac{OC} examples in fading enviroments, clarifying the contribution of \theoremref{Theorem:RealAndImaginaryParts} to deriving closed-form expressions. The \ac{OP} of a communications system signalling in Nakagami-\emph{m} fading environments is given by \cite[Eq.~(33)]{BibYilmazAlouiniTCOM2012}
\begin{equation}\label{Eq:ACCForNakagamiFadingChannels}
    C_{avg}\left(\gammabar\right)=
        \frac{1}{\Gamma(m)}
            \MeijerG[left]{3,1}{2,3}{\frac{m}{\gammabar}}{0,1}{m,0,0},
\end{equation}
where the parameter $m$ denotes the fading figure (i.e., diversity order) ranging from $1/2$ to
$\infty$\cite[Section~2.2.1.4]{BibAlouiniBook}. According to \theoremref{Theorem:RelationshipBetweenAOPAndACC}, we need to have $\ImagPart{C_{avg}\left(-\gammabar\right)}$, and it can be found with the aid of \theoremref{Theorem:RealAndImaginaryParts}, where we derive the \ac{MT} of \eqref{Eq:ACCForNakagamiFadingChannels} using \cite[Eq. (2.1.3)]{BibKilbasSaigoBook} and \cite[Eq. (2.9.1)]{BibKilbasSaigoBook}, that is $\MellinTransform{\gammabar}{C_{avg}\left(\gammabar\right)}{s}=\allowbreak{-}{m}^{s}\Gamma(-s)\Gamma(s)\Gamma(m-s)/\Gamma(m)$, whose substitution into \eqref{Eq:RelationshipBetweenImaginaryPartAPMAndAPM} yields 
\begin{equation}
    \ImagPart{C_{avg}\left(-\gammabar\right)}=
    -\pi\,
    \InvMellinTransform{s}{
        {m}^{s}\frac{\Gamma(m-s)}{s\,\Gamma(m)}
    }{\gammabar},
\end{equation}
where using \cite[Eq. (8.4.16/4)]{BibPrudnikovBookVol3} results in $\ImagPart{C_{avg}\left(-\gammabar\right)}=\pi{\Gamma(m,{m}/{\gammabar})}/{\Gamma(m)}$. Finally, substituting this result into \eqref{Eq:RelationshipBetweenAOPAndACC} yields the well-known \ac{OP} results\cite{BibGoldsmithBook,BibProakisBook,BibAlouiniBook},
\begin{equation}
    \label{Eq:AOPForNakagamiFadingChannels}
    P_{out}\left(\gammabar;\gamma_{th}\right)=
        1-\frac{\Gamma(m,{m}{\gamma_{th}}/{\gammabar})}{\Gamma(m)},
\end{equation}
as expected. The \ac{OC} $\mathcal{C}_{out}(\gammabar;C_{th})$ is easily obtained using \eqref{Eq:AOPForNakagamiFadingChannels} and \eqref{Eq:AOCDefinitionB} together. 

\subsection{Obtaining \ac{SNR} distribution Using Exact \ac{ACC} Expressions}
While the \ac{PDF} of \acp{SNR} distribution, which is often referred to as sample (descriptive) statistics, reveals the relative likelihood of any sample in a continuum occurring, \acp{APM} provide a determination of what is most likely to be correct\,/\,incorrect with the evidence of sample statistics. The notion commonly followed in the literature is to obtain the \acp{APM} using sample statistics such as the \ac{PDF}, \ac{CDF} and \ac{MGF} of \ac{SNR} distribution. conversely, to the best of knowledge, how to obtain sample statistics from \acp{APM} has so far not been explored in the literature. In the following, we demonstrate that the exact (not approximate) \ac{ACC} expression of any communications system is enough to determine the \ac{PDF} of the \ac{SNR} distribution to which overall information transmission is subjected. 

\begin{theorem}[\ac{PDF} of \ac{SNR} distribution using \ac{ACC}]
\label{Theorem:SNRPDFUsingACC}
Let $\mathcal{C}_{avg}(\gammabar)$ be the exact \ac{ACC} of a wireless communication system whose \ac{SNR} distribution follows the \ac{PDF} $f_{\gamma}(r;\gammabar)$. Accordingly, the \ac{PDF} $f_{\gamma}(r;\gammabar)$ is given by 
\begin{equation}\label{Eq:SNRPDFUsingACC}
f_{\gamma}(r;\gammabar)=-
    \frac{1}{\pi}
        \Im\biggl\{
            \frac{\partial}{\partial{r}}\,
                \mathcal{C}_{avg}\left(-\frac{\gammabar}{r}\right)
        \biggr\},
\end{equation}
defined over $r\in\mathbb{R}^{+}$. 
\end{theorem}

\begin{proof}
Note that, for the \ac{SNR} distribution, which is denoted by $\gamma$, the \ac{PDF} $f_{\gamma}(r)$ and the \ac{CDF} $F_{\gamma}(r)$ are defined as
\setlength\arraycolsep{1.4pt}
\begin{eqnarray}
    \label{Eq:PDFDefinition}
    f_{\gamma}(r)&=&
        \Expected{\DiracDelta{r-\gamma}},\\
    \label{Eq:CDFDefinition}
    F_{\gamma}(r)&=&
        \Expected{\HeavisideTheta{r-\gamma}},
\end{eqnarray}
respectively. Using \cite[Eq.~(19.1.3.1/3)]{BibJeffreyHuiHuiBook}, we rewrite the \ac{PDF} $f_{\gamma}(r)$ in terms of the \ac{CDF} $F_{\gamma}(r)$, that is
\begin{equation}\label{Eq:RelationshipBetweenPDFAndCDF}
    f_{\gamma}(r)=\frac{\partial}{\partial{r}}F_{\gamma}(r).
\end{equation}
Finally, noticing $F_{\gamma}(\gamma_{th})\trigeq\mathcal{P}_{out}(\gammabar;\gamma_{th})$ from \eqref{Eq:AOPDefinition} and \eqref{Eq:CDFDefinition}, and subsequently substituting \eqref{Eq:RelationshipBetweenPDFAndCDF} into \eqref{Eq:RelationshipBetweenAOPAndACC}, we deduce \eqref{Eq:SNRPDFUsingACC}, which proves \theoremref{Theorem:SNRPDFUsingACC}.
\end{proof}

In the literature of communications theory, there are several studies and publications about \ac{APM} analyses that widely use two approaches; one of which is the \ac{PDF}-based approach\cite[and references therein]{BibProakisBook} that requires an exact or approximated \ac{PDF} of \ac{SNR} distribution. The other one is the \ac{MGF}-based approach\cite[and references therein]{BibAlouiniBook} that requires an exact or approximated \ac{MGF} of \ac{SNR} distribution for the APM analyses of a diversity receiver with an arbitrary number of diversity branches over a variety of fading channels. In addition to these two approaches, with the aid of \theoremref{Theorem:SNRPDFUsingACC}, we propose in the following theorem a novel approach which we call the \emph{\ac{CC}-based approach}. Moreover, we show that the exact (non-approximate) \ac{ACC} is sufficient and enough to achieve all \ac{APM} analyses, especially without requiring the statistical knowledge (e.g., \ac{PDF}, \ac{CDF}, \ac{MGF}, or moments) of the \ac{SNR} distribution.

\begin{theorem}[\ac{CC}-Based Performance Analysis]
\label{Theorem:CCBasedPerformanceAnalysis}
Let $\mathcal{H}(\gamma)$ be a \ac{PM} according to the \ac{SNR} distribution $\gamma$. The corresponding \ac{APM} defined as $\mathcal{H}_{avg}(\gammabar)\trigeq\Expected{\mathcal{H}(\gamma)}$ is given by 
\begin{equation}\label{Eq:CCBasedPerformanceAnalysis}
\mathcal{H}_{avg}(\gammabar)=
	\mathcal{H}\left(0\right)+
	\frac{1}{\pi}\int_{0}^{\infty}
			\mathcal{H}'\!\left(r\right)
				\Im\biggl\{\mathcal{C}_{avg}\left(\!-\frac{\gammabar}{r}\right)\biggr\}d{r},
\end{equation}
for an average \ac{SNR} $\gammabar$, where
$\mathcal{H}'\!\left(\gamma\right)=\frac{\partial}{\partial\gamma}\mathcal{H}\left(\gamma\right)$.
\end{theorem}

\begin{proof}
Using the \ac{PDF}-based approach when the \ac{PDF} $f_{\gamma}(r;\gammabar)$ is expressed in closed-form expression, we rewrite the \ac{APM} $\mathcal{H}_{avg}(\gammabar)$ as follows
\begin{equation}
    \mathcal{H}_{avg}(\gammabar)=
        -\int_{0}^{\infty}\mathcal{H}\left(r\right)\,
            \frac{\partial}{\partial{r}}\widehat{\mathcal{P}}_{out}(\gammabar;r)\,d{r},
\end{equation}
where $\widehat{\mathcal{P}}_{out}(\gammabar;\gamma_{th})$ is the complementary \ac{OP}. Consequently, using integration-by-part \cite[Eq.~(5.3.11)]{BibZwillingerBook} and therein noticing $\lim_{r\rightarrow{0}}\widehat{\mathcal{P}}_{out}(\gammabar;r)\,\mathord{=}\,{1}$ and $\lim_{r\rightarrow\infty}\widehat{\mathcal{P}}_{out}(\gammabar;r)\,\mathord{=}\,{0}$, we obtain \eqref{Eq:CCBasedPerformanceAnalysis}, which proves \theoremref{Theorem:CCBasedPerformanceAnalysis}. 
\end{proof}

For the analytical accuracy and correctness of \theoremref{Theorem:CCBasedPerformanceAnalysis}, let us consider an extreme example of the \ac{ABER} $\mathcal{E}_{avg}(\gammabar)$ of a communications system signalling over cascaded \ac{GNM} fading channels\cite{BibYilmazAlouiniGLOBECOM2009}. The corresponding \ac{ACC} is given by \cite[Eq. (33)]{BibYilmazAlouiniGLOBECOM2009}   
\begin{equation}\label{Eq:ACCForCascadedGNMfadingChannels}
\!\!\!\!\mathcal{C}_{avg}(\gammabar)=
    \FoxH[left]{N+2,1}{N+2,N+2}
        {\frac{1}{\gammabar}{\prod\limits_{\ell=1}^{N}{{\beta_{\ell}}}}}
            {(0,1),(1,1),\Psi^{(0)}_{N}}
                {\Psi^{(1)}_{N},(0,1),(0,1)},\!\!
\end{equation}
where $\FoxHDefinition{m,n}{p,q}{\cdot}$ denotes the Fox's H' function\cite[Eq. (1.1.1)]{BibKilbasSaigoBook},\cite[Eq. (8.3.1/1)]{BibPrudnikovBookVol3}, $N$ denotes the hop number in the transmission, and  $\Psi^{(n)}_{N}\!=\!(m_{1},{n}/{\xi_{1}}),(m_{2},{n}/{\xi_{2}}),\allowbreak\ldots,(m_{N},{n}/{\xi_{N}})$. In addition, for all $\ell\in\{1,2,\ldots,N\}$, $\beta_{\ell}=\Gamma(m_{\ell}+1/\xi_{\ell})/\Gamma(m_{\ell})$, where $m_{\ell}$ and $\xi_{\ell}$ denote the fading figure (i.e., the diversity order) and the shape parameter of the $\ell$th hop, respectively.~With~the~aid of \theoremref{Theorem:CCBasedPerformanceAnalysis}, we do not necessitate the statistical knowledge of \ac{SNR} distribution for the \ac{ABER} analysis. As such, using the closed-form \ac{ACC} that is given in \eqref{Eq:ACCForCascadedGNMfadingChannels}, we can write the \ac{ABER} $\mathcal{E}_{avg}(\gammabar)$ as follows
\begin{equation}\label{Eq:ABEPUsingACCForCascadedGNMfadingChannels}
\mathcal{E}_{avg}(\gammabar)=
	\mathcal{E}\left(0\right)+
	\frac{1}{\pi}\int_{0}^{\infty}
			\mathcal{E}'\!\left(r\right)
				\Im\biggl\{\mathcal{C}_{avg}\left(\!-\frac{\gammabar}{r}\right)\biggr\}d{r},
\end{equation}
where $\mathcal{E}\left(\gamma\right)$ is the \ac{PM} defined by  $\mathcal{E}\left(\gamma\right)\!=\!{\Gamma\left(b,a\,\gamma\right)}/{\Gamma\left(b\right)}/2$ with parameters $a,b\!\in\!\{1/2,1\}$ (see \secref{Section:RelationshipBetweenACCandABEP:BinaryModulationSchemes} for further details). We have $\lim_{\gamma\rightarrow{0}}\mathcal{E}(\gamma)\!=\!1/2$ and $\lim_{\gamma\rightarrow{\infty}}\mathcal{E}(\gamma)\!=\!0$. Further, using \cite[Eq. (06.06.20.0003.01)]{BibWolfram2010Book} and \cite[Eqs. (1.1.1) and (2.5.7)]{BibKilbasSaigoBook}, we have 
\begin{equation}\label{Eq:WojnarBEPDerivative}
    \frac{\partial}{\partial\gamma}\mathcal{E}\left(\gamma\right)
        =\frac{a}{2\Gamma(b)}
            \FoxH[right]{1,0}{0,1}{a\gamma}{\emptycoefficient}{(b-1,1)},
\end{equation}
In addition, we obtain the \ac{MT} of \eqref{Eq:ACCForCascadedGNMfadingChannels} using the \ac{MT} of Fox's H function\cite[Eq. (8.3.1/1)]{BibPrudnikovBookVol1} and \cite[Eq. (8.3.1/1)]{BibKilbasSaigoBook}. Thereafter, substituting it into \eqref{Eq:RelationshipBetweenImaginaryPartAPMAndAPM}, we obtain 
\begin{equation}\label{Eq:ImagNegACCForCascadedGNMfadingChannels}
\!\!\!\!\Im\Bigl\{\mathcal{C}_{avg}(-\gammabar)\Bigr\}=
    \pi\FoxH[left]{N,N+1}{1,N+1}
            {\frac{1}{\bar{\gamma}}{\prod\limits_{\ell=1}^{N}{{\beta_{\ell}}}}}
                {(0,1),\Psi^{(0)}_{N}}
                    {\Psi^{(1)}_{N},(0,1)},\!\!
\end{equation}
Finally, substituting both \eqref{Eq:WojnarBEPDerivative} and \eqref{Eq:ImagNegACCForCascadedGNMfadingChannels} into \eqref{Eq:ABEPUsingACCForCascadedGNMfadingChannels} and therein employing Mellin's convolution of Fox's H functions\cite[Eq. (2.8.11)]{BibKilbasSaigoBook}, we obtain the \ac{ABER} $\mathcal{E}_{avg}(\gammabar)$ for binary modulation schemes over cascaded \ac{GNM} fading channels as 
\begin{equation}
\label{Eq:ABEPForCascadedGNMfadingChannels}
\!\!\!\!\mathcal{E}_{avg}(\gammabar)=
        \frac{1}{2}
            \FoxH[right]{N,2}{N+2,N+2}
                {\frac{1}{a\bar{\gamma}}\prod\limits_{\ell=1}^{N}{\beta_{\ell}}}
                    {(1-b,1),(1,1),\Psi^{(0)}_{N}}
                        {\Psi^{(1)}_{N},(0,1),(1-b,0)},\!\!
\end{equation}
which is in perfect agreement with \cite[Eq. (40)]{BibYilmazAlouiniGLOBECOM2009} as expected. This example clearly demonstrates how we can benefit from \ac{ACC} performance to obtain other \acp{APM} without knowing the \ac{SNR} distribution. The number of such examples can easily be increased by considering other closed-form \ac{ACC} expressions available in the literature.

\section{Conclusion}
\label{Section:Conclusion}
In this article, we first recommend the use of \ac{LT} to facilitate self-similarity (scale invariance) and therefrom introduce the \ac{LDS} spectrums to identify the similarity between any two \acp{APM}. Secondly, we propose a tractable approach, which we call \ac{LT}-based approach for performance analysis, to establish a relationship between any two \acp{APM}. As such, we demonstrate how to compute one \ac{APM} using the other \ac{APM}, especially without needing the statistical knowledge (such as \ac{PDF}, \ac{CDF}, \ac{MGF} and moments) of \ac{SNR} distribution and knowing the broadest \ac{SNR} settings. 
    
To the best of our knowledge, the literature has currently no answer on how we determine \ac{ACC}~either~empirically~or~experimentally without using the statistical knowledge of SNR distribution. As regards an application of our \ac{LT}-based approach, we propose a relationship in which we can predict the \ac{ACC} of any communications system using its \ac{ABER} performance that we can empirically measure without the need for the statistical knowledge of \ac{SNR} distribution and all the \ac{SNR} settings. 

In addition, we show for the first time in the literature how to obtain sample statistics (such as \ac{PDF}, \ac{CDF} and \ac{MGF} of \ac{SNR} distribution) from \acp{APM}, which changes the playground of performance analysis in the field of wireless communications. We propose that both \ac{OP} and \ac{OC} of any communications system can be obtained by using its exact \ac{ACC} performance. To the best of our knowledge, this relationship has also not been yet reported in the literature. In addition, we introduce a novel approach, which we call the \emph{\ac{CC}-based approach} to perform any \ac{APM} analysis using the exact \ac{ACC} expressions. 
    
Finally, considering some extreme examples, we illustrate the usages and usefulness of our newly proposed relationships. Performing Monte-Carlo simulations, we validate their accuracy and consistency.

\section*{Acknowledgment}
The author thanks Prof. Dr. Mohamed-Slim Alouini of King Abdullah University of Science and Technology (KAUST) for his careful reading of an earlier version of the~article, as well as the Editor and the anonymous Reviewers~for~their~insightful and informative comments that strengthened~the~article.
\ifCLASSOPTIONcaptionsoff
\newpage
\fi


\ifCLASSOPTIONcaptionsoff
\newpage
\fi

\ifCLASSOPTIONtwocolumn
\bibliography{IEEEabrv,yilmaz_relation_between_capacity_and_bit_error_rate}
\bibliographystyle{IEEEtran} 
\ifCLASSOPTIONtwocolumn
\pagebreak[4]
\fi
\else
\bibliography{IEEEfull,yilmaz_relation_between_capacity_and_bit_error_rate}

\begin{thebibliography}{10}
\providecommand{\url}[1]{#1}
\csname url@samestyle\endcsname
\providecommand{\newblock}{\relax}
\providecommand{\bibinfo}[2]{#2}
\providecommand{\BIBentrySTDinterwordspacing}{\spaceskip=0pt\relax}
\providecommand{\BIBentryALTinterwordstretchfactor}{4}
\providecommand{\BIBentryALTinterwordspacing}{\spaceskip=\fontdimen2\font plus
\BIBentryALTinterwordstretchfactor\fontdimen3\font minus
  \fontdimen4\font\relax}
\providecommand{\BIBforeignlanguage}[2]{{%
\expandafter\ifx\csname l@#1\endcsname\relax
\typeout{** WARNING: IEEEtran.bst: No hyphenation pattern has been}%
\typeout{** loaded for the language `#1'. Using the pattern for}%
\typeout{** the default language instead.}%
\else
\language=\csname l@#1\endcsname
\fi
#2}}
\providecommand{\BIBdecl}{\relax}
\BIBdecl

\bibitem{BibDongningGuoShamaiVerduTIT2005}
{Dongning Guo}, S.~{Shamai}, and S.~{Verdu}, ``Mutual information and minimum
  mean-square error in {G}aussian channels,'' \emph{{IEEE} Trans. Inf. Theory},
  vol.~51, no.~4, pp. 1261--1282, 2005.

\bibitem{BibProakisBook}
J.~{Proakis}, \emph{Digital Communications}, 4th~ed.\hskip 1em plus 0.5em minus
  0.4em\relax McGraw-Hill Science/Engineering/Math, Aug. 2000.

\bibitem{BibAlouiniBook}
M.~K. {Simon} and M.-S. {Alouini}, \emph{Digital Communication over Fading
  Channels}, 2nd~ed.\hskip 1em plus 0.5em minus 0.4em\relax John Wiley \& Sons,
  Inc., 2005.

\bibitem{BibGoldsmithBook}
A.~{Goldsmith}, \emph{Wireless Communications}.\hskip 1em plus 0.5em minus
  0.4em\relax Cambridge University Press, Aug.8, 2005.

\bibitem{BibSimonAlouiniProcIEEE1998}
M.~K. {Simon} and M.-S. Alouini, ``A unified approach to the performance
  analysis of digital communication over generalized fading channels,''
  \emph{Proceedings of the {IEEE}}, vol.~86, no.~9, pp. 1860--1877, Sep. 1998.

\bibitem{BibAnnamalaiTellamburaBhargavaTCOM2005}
A.~{Annamalai}, C.~{Tellambura}, and V.~K. {Bhargava}, ``A general method for
  calculating error probabilities over fading channels,'' \emph{{IEEE} Trans.
  Commun.}, vol.~53, no.~5, pp. 841--852, May 2005.

\bibitem{BibYilmazAlouiniICC2012}
F.~{Yilmaz} and M.-S. {Alouini}, ``A novel ergodic capacity analysis of
  diversity combining and multihop transmission systems over generalized
  composite fading channels,'' in \emph{proc. {IEEE} International Conference
  on Communications (ICC 2012), Ottawa, Canada}, June 2012, pp. 4605--4610.

\bibitem{BibYilmazAlouiniTCOM2012}
------, ``A unified {MGF}-based capacity analysis of diversity combiners over
  generalized fading channels,'' \emph{{IEEE} Trans. Commun.}, vol.~60, no.~3,
  pp. 862--875, Mar. 2012.

\bibitem{BibYilmazAlouiniTCOM2012UnifiedExpression}
------, ``A novel unified expression for the capacity and bit error probability
  of wireless communication systems over generalized fading channels,''
  \emph{{IEEE} Trans. Commun.}, vol.~60, no.~7, pp. 1862--1876, 2012.

\bibitem{BibYilmazAlouiniWCL2012}
------, ``On the computation of the higher-order statistics of the channel
  capacity over generalized fading channels,'' \emph{{IEEE Wireless
  Communications Letters}}, vol.~1, no.~6, pp. 573--576, 2012.

\bibitem{BibYilmazAlouiniSPAWC2012}
------, ``Novel asymptotic results on the high-order statistics of the channel
  capacity over generalized fading channels,'' in \emph{{IEEE Inter. Workshop
  on Signal Processing Advances in Wireless Commun. (SPAWC 2012)}}, 2012, pp.
  389--393.

\bibitem{BibYilmazTUBITAK2019}
F.~{Yilmaz}, ``On the asymptotic analysis of the high-order statistics of the
  channel capacity over generalized fading channels,'' \emph{{Turkish Journal
  of Electrical Engineering \& Computer Sciences}}, 2019, {A}ccepted for
  publication.

\bibitem{BibLampertiAMS1962}
J.~{Lamperti}, ``Semi-stable stochastic processes,'' \emph{Transactions of the
  American Mathematical Society}, vol. 104, pp. 62--78, 1962.

\bibitem{BibShannonBSTJ1948}
C.~E. {Shannon}, ``A mathematical theory of communication,'' \emph{Bell System
  Tech. Journal}, vol.~27, pp. 379--423,~623--656, Julyand Oct. 1948,
  {A}vail\-able at http://cm.bell-labs.com/cm/ms/what/shannonday/paper.html.

\bibitem{BibShannon1949}
------, ``Communications in the presence of noise,'' in \emph{Proc. {IRE}},
  1949, pp. 10--21.

\bibitem{BibShannonWeaverUIP1949}
C.~E. {Shannon} and W.~{Weaver}, \emph{{The Mathematical Theory of
  Communication}}.\hskip 1em plus 0.5em minus 0.4em\relax Urbana, IL, USA:
  University of Illinois Press, 1949.

\bibitem{BibZwillingerBook}
D.~{Zwillinger}, \emph{{CRC} {S}tandard {M}athematical {T}ables and
  {F}ormulae}, 31st~ed.\hskip 1em plus 0.5em minus 0.4em\relax Boca Raton, FL:
  Chapman \& Hall/CRC, 2003.

\bibitem{BibMandelbrotBook}
B.~B. {Mandelbrot}, \emph{The fractal geometry of nature}, 1st~ed.\hskip 1em
  plus 0.5em minus 0.4em\relax W.H. Freeman, Aug 1982.

\bibitem{BibChampeneyBook}
D.~C. {Champeney}, \emph{A handbook of Fourier theorems}.\hskip 1em plus 0.5em
  minus 0.4em\relax Cambridge University Press, 1987.

\bibitem{BibSneddonFourierTrasforms1995Book}
I.~N. {Sneddon}, \emph{{F}ourier {T}ransforms}, ser. Dover Books on
  Mathematics.\hskip 1em plus 0.5em minus 0.4em\relax Dover Publications, 1995.

\bibitem{BibOberhettingerBook}
F.~{Oberhettinger}, \emph{Tables of Mellin Transforms}.\hskip 1em plus 0.5em
  minus 0.4em\relax Springer-Verlag, 1974.

\bibitem{BibPoularikasBook}
A.~D. {Poularikas}, \emph{The Transforms and Applications Handbook}, ser. The
  Electrical Engineering Handbook Series.\hskip 1em plus 0.5em minus
  0.4em\relax {CRC} {P}ress, 2000.

\bibitem{BibKilbasSaigoBook}
A.~{Kilbas} and M.~{Saigo}, \emph{{H}-Transforms: Theory and
  Applications}.\hskip 1em plus 0.5em minus 0.4em\relax Boca Raton, FL: CRC
  Press LLC, 2004.

\bibitem{BibMathaiSaxenaHauboldBook}
A.~M. {Mathai}, R.~K. {Saxena}, and H.~J. {Haubold}, \emph{The {H}-{F}unction:
  {T}heory and {A}pplications}.\hskip 1em plus 0.5em minus 0.4em\relax New York
  Dordrecht Heidelberg London: Springer, 2010.

\bibitem{BibDavisBook1975}
P.~J. {Davis}, \emph{{I}nterpolation and {a}pproximation}.\hskip 1em plus 0.5em
  minus 0.4em\relax {Courier Corporation}, 1975.

\bibitem{BibPressTeukolskyNumericalRecipes1992Book}
W.~H. {Press}, S.~A. {Teukolsky}, W.~T. {Vetterling}, and B.~P. {Flannery},
  \emph{{N}umerical {R}ecipes in {C} - {T}he {A}rt of {S}cientific
  {C}omputing}, 2nd~ed.\hskip 1em plus 0.5em minus 0.4em\relax Cambridge, U.K.:
  Cambridge University., 1992.

\bibitem{BibWong2010Book}
S.~S.~M. {Wong}, \emph{{C}omputational {M}ethods in {P}hysics and
  {E}ngineering}, 2nd~ed.\hskip 1em plus 0.5em minus 0.4em\relax Upper Saddle
  River, NJ: World Scientific, 1997.

\bibitem{BibLeeTVT1990}
W.~C.~Y. {Lee}, ``Estimate of channel capacity in {R}ayleigh fading
  environment,'' \emph{{IEEE} Trans. Veh. Technol.}, vol.~39, no.~3, pp.
  187--189, Aug. 1990.

\bibitem{BibEricsonTIT1970}
T.~{Ericson}, ``A {G}aussian channel with slow fading ({C}orresp.),''
  \emph{{IEEE} Trans. Inf. Theory}, vol.~16, no.~3, pp. 353--355, May 1970.

\bibitem{BibOzarowShamaiWynerTVT1994}
L.~H. {Ozarow}, S.~{Shamai}, and A.~D. {Wyner}, ``Information theoretic
  considerations for cellular mobile radio,'' \emph{{IEEE} Trans. Veh.
  Technol.}, vol.~43, no.~2, pp. 359--378, May 1994.

\bibitem{BibGoldsmithVaraiyaTIT1997}
A.~J. {Goldsmith} and P.~P. {Varaiya}, ``Capacity of fading channels with
  channel side information,'' \emph{{IEEE} Trans. Inf. Theory}, vol.~43, no.~6,
  pp. 1986--1992, Nov. 1997.

\bibitem{BibAlouiniGoldsmithTVT1999}
M.-S. {Alouini} and A.~J. {Goldsmith}, ``Capacity of {R}ayleigh fading channels
  under different adaptive transmission and diversity-combining techniques,''
  \emph{{IEEE} Trans. Veh. Technol.}, vol.~48, no.~4, pp. 1165--1181, July
  1999.

\bibitem{BibAlouiniGoldsmithWPC2000}
------, ``Adaptive modulation over {N}akagami fading channels,''
  \emph{{W}ireless {P}ersonal {C}ommun.}, vol.~13, pp. 119--143, 2000.

\bibitem{BibViswanathanTIT1999}
H.~{Viswanathan}, ``Capacity of {M}arkov channels with receiver {CSI} and
  delayed feedback,'' \emph{{IEEE} Trans. Inf. Theory}, vol.~45, no.~2, pp.
  761--771, Mar. 1999.

\bibitem{BibJayaweeraPoorTIT2003}
S.~K. {Jayaweera} and H.~V. {Poor}, ``Capacity of multiple-antenna systems with
  both receiver and transmitter channel state information,'' \emph{{IEEE}
  Trans. Inf. Theory}, vol.~49, no.~10, pp. 2697--2709, Oct. 2003.

\bibitem{BibYuanZhangTepedelenliogluTIT2012}
{Yuan Zhang} and C.~{Tepedelenlioglu}, ``Asymptotic capacity analysis for
  adaptive transmission schemes under general fading distributions,''
  \emph{{IEEE} Trans. Inf. Theory}, vol.~58, no.~2, pp. 897--908, Feb. 2012.

\bibitem{BibAlouiniGoldsmithTCOM1999}
M.-S. {Alouini} and A.~J. {Goldsmith}, ``A unified approach for calculating
  error rates of linearly modulated signals over generalized fading channels,''
  \emph{{IEEE} Trans. Commun.}, vol.~47, no.~9, pp. 1324--1334, Sep. 1999.

\bibitem{BibWojnarTCOM1986}
A.~{Wojnar}, ``Unknown bounds on performance in {N}akagami channels,''
  \emph{{IEEE} Trans. Commun.}, vol.~34, no.~1, pp. 22--24, Jan. 1986.

\bibitem{BibAbramowitzStegunBook}
M.~{Abramowitz} and I.~A. {Stegun}, \emph{Handbook of Mathematical Functions
  with Formulas, Graphs, and Mathematical Tables}, 9th~ed.\hskip 1em plus 0.5em
  minus 0.4em\relax New York: Dover Publications, 1972.

\bibitem{BibPrudnikovBookVol3}
A.~P. {Prudnikov}, Y.~A. {Brychkov}, and O.~I. {Marichev}, \emph{Integral and
  Series: Volume 3, More Special Functions}.\hskip 1em plus 0.5em minus
  0.4em\relax CRC Press Inc., 1990.

\bibitem{BibGradshteynRyzhikBook}
I.~S. {Gradshteyn} and I.~M. {Ryzhik}, \emph{Table of {I}ntegrals, {S}eries,
  and {P}roducts}, 5th~ed.\hskip 1em plus 0.5em minus 0.4em\relax San Diego,
  CA: Academic Press, 1994.

\bibitem{BibWolfram2010Book}
{Wolfram Research}, \emph{{M}athematica {E}dition: {V}ersion 8.0}.\hskip 1em
  plus 0.5em minus 0.4em\relax Champaign, Illinois: Wolfram Research, Inc.,
  2010.

\bibitem{BibPrudnikovBookVol2}
A.~P. {Prudnikov}, Y.~A. {Brychkov}, and O.~I. {Marichev}, \emph{Integral and
  Series: Volume 2, Special Functions}.\hskip 1em plus 0.5em minus 0.4em\relax
  CRC Press Inc., 1990.

\bibitem{BibAaloPiboongungonIskander2005}
V.~{Aalo}, T.~{Piboongungon}, and C.-D. {Iskander}, ``Bit-error rate of binary
  digital modulation schemes in generalized {G}amma fading channels,''
  \emph{{IEEE} Commun. Lett.}, vol.~9, no.~2, pp. 139--141, Feb. 2005.

\bibitem{BibPrudnikovBookVol4}
A.~P. {Prudnikov}, Y.~A. {Brychkov}, and O.~I. {Marichev}, \emph{Integral and
  Series: Volume 4, Direct Laplace Transforms}.\hskip 1em plus 0.5em minus
  0.4em\relax CRC Press Inc., 1990.

\bibitem{BibJeffreyHuiHuiBook}
A.~{Jeffrey} and H.-H. {Dai}, \emph{Handbook of Mathematical Formulas and
  Integrals}, 4th~ed.\hskip 1em plus 0.5em minus 0.4em\relax Academic Press,
  2008.

\bibitem{BibYilmazAlouiniGLOBECOM2009}
F.~{Yilmaz} and M.-S. {Alouini}, ``Product of the powers of generalized
  nakagami-m variates and performance of cascaded fading channels,'' in
  \emph{proc. {IEEE} Global Communications Conference (GLOBECOM 2009),
  Honolulu, Hawaii, USA}, Nov. 30-Dec. 4 2009.

\bibitem{BibPrudnikovBookVol1}
A.~P. {Prudnikov}, Y.~A. {Brychkov}, and O.~I. {Marichev}, \emph{Integral and
  Series: Volume 1, Elementary Functions}.\hskip 1em plus 0.5em minus
  0.4em\relax CRC Press Inc., 1990.

\end{thebibliography}
\bibliographystyle{IEEEtran} 
\fi
\phantomsection
\vfill
\ifCLASSOPTIONtwocolumn
\begin{IEEEbiography}[{\includegraphics[width=1in,height=1.25in,clip,keepaspectratio]{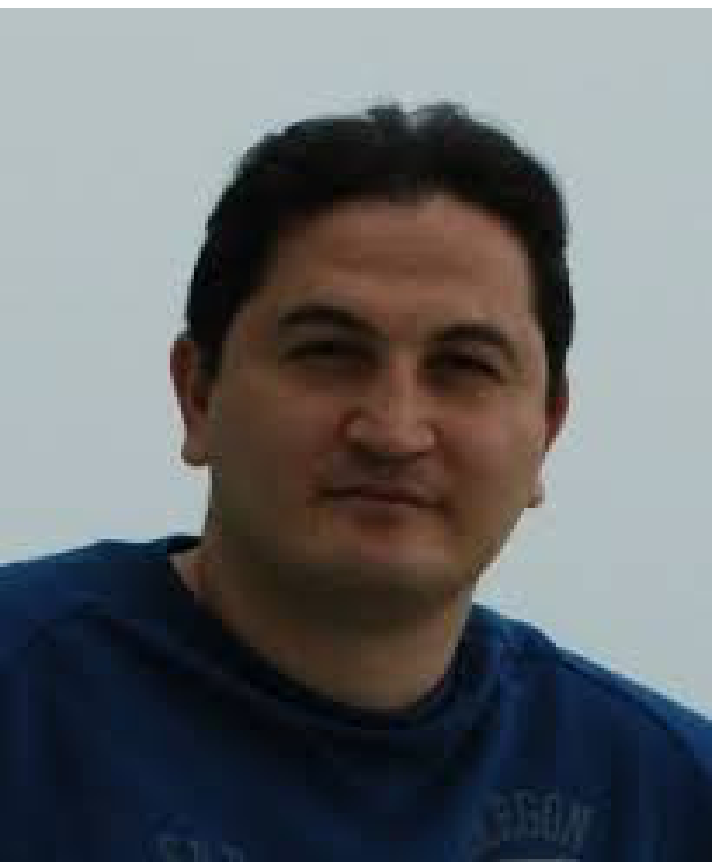}}]{Ferkan Yilmaz} (SM'2004--M'2009) received the B.Sc. degree in electronics and communications engineering from Y{\i}ld{\i}z Technical University, Turkey, in 1997. He was ranked the first among all undergraduate students who graduated from the Department of Electronics and Communication Engineering (ECE), and the second among all undergraduate students who graduated in 1997 from Y{\i}ld{\i}z Technical University. He received his M.Sc. degree in the field of electronics and communications from Istanbul Technical University, Turkey, in 2002, and the Ph.D. degree in field of wireless communications from Gebze Institute of Technology (GYTE) in 2009. He was awarded for the best Ph.D. Thesis.
 
From 1998 to 2003, he worked as a researcher for National Research Institute of Electronics and Cryptology, T\"{U}B\.{I}TAK, Turkey. Dr. Yilmaz is the recipient of \emph{The 1999 Achievement Award} from T\"{U}B\.{I}TAK. 

From 2003 to 2008, he worked as a senior Telecommunications Engineer for Vodafone Technology, Turkey. In November 2008, he was invited by Prof. Dr. M.-S. Alouini to work as a visitor researcher at Texas A{\&}M University (Qatar), and from August 2009 till November 2012, worked as a post-doctoral fellow at King Abdullah University of Science and Technology (KAUST). Between November 2012 and April 2015, he joined as a Senior Telecommunications expert in the Location-based Service Solutions group, Vodafone Technology, and therein developed positioning algorithms and achieved the implementation and assessment of indoor\,/\,outdoor location-aware applications. In 2015, he was with KAUST as a Remote Consultant working on optical communications. 

In 2016, Dr. Yilmaz joined to the Department of Computer Engineering at Y{\i}ld{\i}z Technical University. He has been currently working as a Assistant Professor and also engaged in \emph{Alpharabius Solutions} and \emph{Tachyonic Solutions} startup companies. He served as a TPC member for several IEEE conferences and is a regular reviewer for various IEEE journals. His research interests focus on wireless communication theory (i.e, diversity techniques, cooperative communications, spatial modulation, fading channels), signal processing, stochastic random processes, statistical learning theory, and machine learning and data mining (big data, data analytics).
\end{IEEEbiography} 
\fi

\end{document}